\documentclass[10pt,pra,amsmath,amssymb,aps,superscriptaddress,notitlepage,twocolumn]{revtex4-1}
\usepackage{libertine}

\usepackage[utf8]{inputenc}
\usepackage{amsmath,amsfonts,amssymb,amsthm}
\usepackage{graphicx}
\usepackage{bbm,bbold,comment}
\usepackage{tikz}
\usepackage{enumitem}
\usepackage{color}
\usepackage[normalem]{ulem}
\usepackage{etex}
\newtheorem{theorem}{Theorem}
\newtheorem{proposition}{Proposition}

\newcounter{axiom}
\newcounter{appaxiom}

\newtheorem{lemma}[theorem]{Lemma}
\newtheorem{corollary}[theorem]{Corollary}
\newtheorem{definition}[theorem]{Definition}

\newtheorem{example}[theorem]{Example}

\newtheorem{axm}[axiom]{Axiom}
\newtheorem{appaxm}[appaxiom]{Axiom}
\newtheorem{property}{Property}
\newcommand{\RR}{\mathbb{R}}
\newcommand{\CC}{\mathbb{C}}

\newcommand{\NN}{\mathbb{N}}
\newcommand{\id}{\mathbbm{1}}
\newcommand{\mc}[1]{\mathcal{#1}}

\newcommand{\ket}[1]{|#1\rangle}

\newcommand{\ketbra}[2]{| #1 \rangle \langle #2 |}

\newcommand{\proj}[1]{\vert #1\rangle\!\langle#1 \vert}

\newcommand{\norm}[1]{\left\Vert #1 \right\Vert}

\newcommand{\ro}[1]{\textcolor{black}{#1}}
\newcommand{\je}[1]{\textcolor{black}{#1}}
\newcommand{\he}[1]{{\color{black} #1}}
\definecolor{henrik}{rgb}{0,0,0}

\newcommand{\Tr}{\operatorname{tr}}
\newcommand{\tr}{\Tr}
\newcommand{\fu}{Dahlem Center for Complex Quantum Systems, Freie Universit{\"a}t Berlin, 14195 Berlin, Germany}

\begin{document}
\title{Thermodynamic work from operational principles}
 
\author{R.\ Gallego, J.\ Eisert, and H. Wilming}
\affiliation{\fu}

\begin{abstract}
In recent years we have witnessed a concentrated effort to make sense of thermodynamics for small-scale systems. One of the main difficulties is to capture a suitable notion of work that models realistically the purpose of quantum machines, in an analogous way to the role played, for macroscopic machines, by the energy stored in the idealisation of a lifted weight. Despite of several attempts to resolve this issue by putting forward specific models, these are far from capturing realistically the transitions that a quantum machine is expected to perform. In this work, we adopt a novel strategy by considering arbitrary kinds of systems that one can attach to a quantum thermal machine and seeking for \emph{work quantifiers}. These are functions that measure the value of a transition and generalise the concept of work beyond the model of a lifted weight. We do so by imposing simple operational axioms that any reasonable \emph{work quantifier} must fulfil and by deriving from them stringent mathematical condition with a clear physical interpretation. Our approach allows us to derive much of the structure of the theory of thermodynamics without taking as a primitive the definition of work. We can derive, for any \emph{work quantifier}, a quantitative second law in the sense of bounding the work that can be performed using some non-equilibrium resource by the work that is needed
to create it. We also discuss in detail the role of reversibility and correlations in connection with the second law. Furthermore, we recover the usual identification of work with energy in degrees of freedom with vanishing entropy as a particular case of our formalism. Our mathematical results can be formulated abstractly and are general enough to carry over to other resource theories than quantum thermodynamics.
\end{abstract}
\maketitle

\section{Introduction}\label{sec:problems}

With the advent of highly-controlled experiments with small-scale quantum devices and the technological
perspective to use such devices as \emph{machines}
\cite{Exp1,Exp2,Exp3,Exp4,Exp5} 
it is becoming increasingly important to understand what it 
precisely means for such a machine to {extract work}.

For macroscopic, classical machines, work can be grasped in several equivalent ways. In particular, it can be captured by introducing a \emph{work-storage device}, which can be seen as being modelled by a lifted weight. This model considers a body in a conservative force (a weight), described by a single a deterministic state-variable (the height). Importantly, the body cannot be used as an entropy sink or reservoir. By defining work proportional to the height-difference of the weight, one ensures that work captures a notion of operationally useful energy
and as such reflects the actual purpose of macroscopic thermal machines, which may literally be seen as raising a weight.

{In contrast, the situation is much less clear when considering microscopic thermal machines in physical situations
in which quantum effects are expected to be relevant. 
Even in this regime, one may readily conceive quantum analogues of a lifted weight  \cite{Skrzypczyk2010,Nanomachines,Abergtrullywork,Renes14,Halpern14,Faist15},
but it remains conceptually more challenging to justify such notions as reflecting the behaviour of realistic machines at the nano-scale.
Taking the idea seriously that quantum effects are expected to play a role, one should take into account the possibility that the 
work-storage device $A$ itself may be composed of a few atoms only, and thus thermal fluctuations may well become comparable with the typical energy scales
of the machine. Hence, such a microscopic system $A$ is expected to end up with fluctuating values of energy or act as an entropy sink or reservoir.
It may even be necessary to take coherences into account in a quantum description \cite{Aberg,Piotr14,Faist14,Lostaglio14}. All these features suggest that the idealisation of the lifted weight, \he{while} possible also in the setting where quantum effects are expected to be relevant, may not 
\ro{have the same degree of applicability to model realistic situations as it has in macroscopic thermodynamics.}}

{In this work, we introduce a fresh approach towards addressing the problem of dealing with the notion of work in a general fashion.
We propose to allow for more general classes of work-storage devices---going beyond those considered in Refs.\ \cite{Skrzypczyk2010,Nanomachines,Abergtrullywork,Renes14,Halpern14,Faist15}---and quantify a notion of 
work for each such class of systems. The framework introduced is general enough to recover known notions of work
as particular cases of the formalism.}

Conceptually more importantly still, we will adopt a strictly operational perspective. Rather than aiming at defining \ro{work as an \emph{a priori} given quantity}, we will 
advocate th\ro{e e}nterprise 
to first carefully and precisely state what operational properties any quantity reasonably measuring work
for arbitrary classes of work-storage devices
should satisfy. We cast these requirements in the form of basic operational axioms that we expect a measure of work to fulfil.
%
%We will adopt a 
%strictly operational perspective: we introduce an operational framework to formulate
%the question of measuring work given an arbitrary class of work-storage devices and
%propose basic operational axioms that we expect any reasonable
%measure of work to fulfil. 
{From these elementary axioms, we derive surprisingly stringent and specific conditions to the 
relevant work quantifier.} In this way, we approach the question of defining work from an entirely new angle.

{Our approach builds upon and further develops  ideas from 
\emph{quantum resource theories} \cite{ResourceTheory,Horodecki13,Coecke14,Brandao15,Lidia};
as a consequence of this approach, we obtain several results on resource theories interesting in their own right. In particular, 
our results highlight the role of correlations and so-called \emph{catalysts} \cite{SecondLaw,Ng14,Mueller14}. 
Our approach also draws some inspiration from the \emph{axiomatic approach to thermodynamics} put
forth in the seminal Ref.\ \cite{Lieb13}, in which axiomatics for thermodynamic state transformations is
introduced, even though the object of study here being quite different. }
\bigskip

\section{A motivating example}

Consider the situation where a thermal machine {runs a protocol aiming at work extraction. The protocol will change the state of a system that we label by $A$ and that we refer to as the \emph{work-storage device}. One can think of $A$ as taking the role of a battery that is being charged or a weight that is being lifted.} Suppose that initially the system 
$A$ is in a state with vanishing deterministic energy and finally with deterministic energy $\Delta$. Moreover this change of energy 
is assumed to happen  in each run of the experiment, that is, with unit probability.  {If one associates work \he{to a deterministic change of} energy, as it is usually done in thermodynamics \cite{Jarzynski97,Crooks99,Talkner07}, then we will conclude that} the machine extracts work $\Delta$ with certainty.

{However, whether energy can be directly associated with work or should rather be regarded as heat depends crucially on what is considered a valid work-storage device.}  To see this, we will elaborate on {the following example: Let us consider as work-storage device a} \he{quantum} system $A$ described by a Hamiltonian which has a unique ground state and a $d$-fold degenerate excited energy-level with energy $\Delta$. Then it is possible, {as discussed in Fig.\ \ref{fig:counterex}, to bring $A$ in contact with a heat bath and map the ground state of $A$ to a state with deterministic energy $\Delta$---}namely the maximally mixed state within that subspace. {The energy $\Delta$ was stored deterministically in $A$, but entirely provided by a heat bath, thus it is at least questionable whether it should be associated with heat-like energy instead of 
actual work.} 

{It is easy to appreciate in this example what it is that allows for \ro{such a process}: the work-storage device is acting not only as a storage of deterministic energy, but also as an entropy sink. Hence, basic notions of phenomenological thermodynamics compel one to reject \ro{a system $A$ like the one of Fig \ref{fig:counterex}} as a valid work-storage device. \ro{Importantly, \he{within phenomenological thermodynamics,} imposing that $A$ must be modelled by the idealisation of a lifted weight is very well motivated from an operational perspective.} \ro{The reason \he{is} that it serves to model accurately the purpose of macroscopic thermal machines, which is indeed often equivalent to literally lifting a heavy object attached to a rope. \he{Moreover, the machine should lift the object} independently of the internal state of the object.}
}
%Indeed, within the well understood theory of phenomenological thermodynamics, $A$ is not any kind of system, but one equivalent to a lifted weight, which by no means is allowed to act as an entropy sink. 
%Note also, that it would be of little relevance to question the idealisation of the lifted weight, since it serves to model the purpose of macroscopic thermal machines, which is indeed often %equivalent to literally lifting a heavy object attached to a rope, in such a way the machine should do so independently of the internal state of the object.}
%[I do not understand this sentence.]

Again, a most natural question that emerges is whether a similar reasoning, \ro{advocating that $A$ must be modeled by a lifted weight,} can be applied at the nano or atomic scale where quantum effects are expected to be relevant. That is, as discussed in the introduction, a microscopic thermal machine might be attached to a system $A$ consisting of a few atoms. Hence, the idealisation of the lifted weight---or its microscopic analogues \cite{Skrzypczyk2010,Nanomachines,Abergtrullywork,Renes14,Halpern14,Faist15}---are far from capturing all relevant and realistic transformations that one expects a nano-machine to perform. Of course, we do not claim that the example of Fig.\ \ref{fig:counterex} models a realistic system. Nonetheless, it illustrates the difficulties emerging when work-storage devices are not idealisations of the like of a lifted weight, thus making problematic the identification of work with deterministic energy. 

In the following sections, we will largely overcome those difficulties by dropping the conceptual guide that work must be analogue to the energy that we store in the microscopic version of a lifted weight. We will show that for any classes of allowed systems describing $A$ 
%(regardless or whether they are in any sense analogues of lifted weights) 
one can find \emph{work quantifiers}: functions that, for a given class of systems $A$, {associate to each transition of $A$ a real number. These functions behave analogously to the energy when the class of systems is indeed taken to be a lifted weight: the second-law can be expressed as a limitation to the value this function can take; do not allow for a \emph{perpetuum-mobil\'e} creating positive value of this function without burning resources; and the maximum values of the function are obtained by reversible processes. Furthermore, the common identification of energy of a lifted weight and work is recovered as a particular case of the work-quantifiers. All these properties suggest that we identify and refer to these functions as \emph{work quantifiers}. }

\begin{figure}
\includegraphics{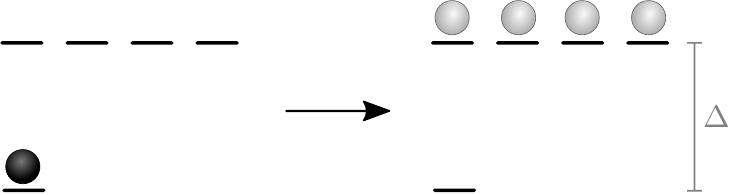}
\caption{{Consider a work-storage device $A$ with a Hamiltonian $H_{d,\Delta}$ having a unique ground state and a $d$-fold degenerate excited level with energy $\Delta$. Let us assume that $A$ is initially in the ground state. We put this system in contact with a heat bath at inverse temperature $\beta$, so that the final state is the Gibbs state with inverse temperature $\beta>0$. This state is such the final probability of being in the ground state is $p_g=1/Z(H_{d,\Delta})$, where $Z(H)$ is the partition function of $H$. For any value of $\Delta$ one can find $d$ such that $p_g$ is arbitrarily close to zero. \he{Using a more sophisticated operation consisting of an energy-conserving unitary on bath and system one can even achieve $p_g=0$ exactly as long as $d>\exp(\beta \Delta)$. This can be seen using techniques of Ref.\ \cite{SecondLaw}.}
This example shows that a deterministic change of deterministic energy can in principle be done by i) only using a single heat bath and ii) bringing the system closer to thermal equilibrium. Note also that the protocol works for every value of $\Delta$ if $d$ is large enough. Hence, it cannot be interpreted as an exponentially suppressed statistical violation of the second law, as the results in Refs.\ \cite{Jarzynski97,Crooks99,Talkner07} are often interpreted.  }}
\label{fig:counterex}
\end{figure}

\section{The operational framework and axioms}

The previous considerations motivate our approach to the problem of
defining work in terms of an operational viewpoint.
We formulate it as a game between two players. The first one is Arthur, who possesses a quantum system which takes the role of the \emph{work-storage device}. 
The system is described by a pair of a quantum state and a Hamiltonian 
\begin{equation}
	p=(\rho,H), 
\end{equation}
referred to as \emph{object}. Typically one assigns certain properties to what is considered a valid work-storage system. For instance, {as discussed in the introduction, in phenomenological thermodynamics the usual demand is to impose that it is a body under a conservative force (such as, again, a  lifted weight). To reiterate, 
complementing notions capturing the idea of a microscopic analogue of a lifted weight have been put forward \he{in the literature \cite{Skrzypczyk2010,Nanomachines,Abergtrullywork,Renes14,Halpern14,Faist15}. However, }here we are precisely interested in considering arbitrary classes of work-storage devices, hence, we will be fully general and encode such constraints by assuming that $p$ belongs to some set $\mc{P}$. \footnote{{Note that the choice of $\mc{P}$ 
has a subjective element, in the same way it amounts to a restriction to take $\mc{P}$ as a lifted weight. It describes a particular choice of systems that we consider valid resources because we can handle them in a given experimental situation. To put an example, one can imagine that a given experimental setup can only handle qubits, systems of bounded entropy or energy, or systems with a fixed Hamiltonian. Our formalism allows to choose $\mc{P}$ so that it encodes each of those situations.}}}. The system is initially described by $p^{(i)} \in \mc{P}$. The second player is 
referred to as Merlin, who has a machine capable of performing transitions between the initial object $p^{(i)}$ to a final object $p^{(f)}\in \mc{P}$. This machine will play the role of the thermal machine or engine, which performs a transformation on the work-storage device.

We assume that the transition $p^{(i)} \rightarrow p^{(f)}$ is performed
in an environment of temperature $T=1/\beta$ (we set $k_B=1$), which we
consider to be fixed throughout the rest of our work. In a resource-theoretic setting, this means that Merlin performs the transitions while having unlimited and free access to arbitrary heat baths at inverse temperature $\beta$. The term ``free'' is here used in the sense of a resource theory, a notion that will be made precise later.

Arthur and Merlin, having agreed on the free character of heat baths and the properties of the work-storage device given by $\mc{P}$, would like to establish a fair way of quantifying the value of a given transition $p^{(i)} \rightarrow p^{(f)}$. That is, they aim at agreeing on a function 
\begin{equation}
	(p^{(i)} \rightarrow p^{(f)}, \beta)\mapsto \mc{W}(p^{(i)} \rightarrow p^{(f)}, \beta) 
\end{equation}
for any $p^{(i)},p^{(f)} \in \mc{P}$, that establishes the prize at which Merlin sells to Arthur the transition that
he has performed, where we will take the convention that $\mc{W}(p^{(i)} \rightarrow p^{(f)}, \beta) \geq 0$ implies that Arthur has to pay to Merlin. The prize of the transition is what we define as \emph{work}, and the function $\mc{W}$ a \emph{work quantifier}. {The notion of ``fair'' prize will encode properties that $\mc{W}$ fulfils when it is identified simply with the energy difference and $\mc{P}$ is a lifted weight. It is precisely in this sense that $\mc{W}$ plays an analogue role of the one played by work in thermodynamics, and what justifies that we refer to it as work quantifier.} 

Since the temperature of the free heat baths is fixed, we also often
write $\mc{W}(p^{(i)} \rightarrow p^{(f)})$ for $\mc{W}(p^{(i)} \rightarrow p^{(f)},\beta)$.
Apart from the agreement on the free heat baths at inverse temperature $\beta$, the work value has to be established solely on the basis of which transition $p^{(i)} \rightarrow p^{(f)}$ is performed by Merlin, without any assumption or restriction on the internal mechanism of Merlin's device. {This is also a property inherited from the usual notion of work within phenomenological thermodynamics, where the work can be quantified by looking only at the 
initial and final state (height) of the lifted weight. 
\je{Lastly, we would like to stress that the language making reference to players such as  Arthur and 
Merlin captures the usual thermodynamic setting and bounds the very same
quantities usually under consideration in thermodynamics, as explained in Sec.\ \ref{sec:notionsofwork}. 
Yet, the novel operational framework introduced in this work is most transparently 
stated in such a language, as a pedagogical tool inspired by common notions of  interactive proof systems in theoretical computer science.}

\subsection{Free catalytic transitions}

\begin{figure*}[htb]
\includegraphics{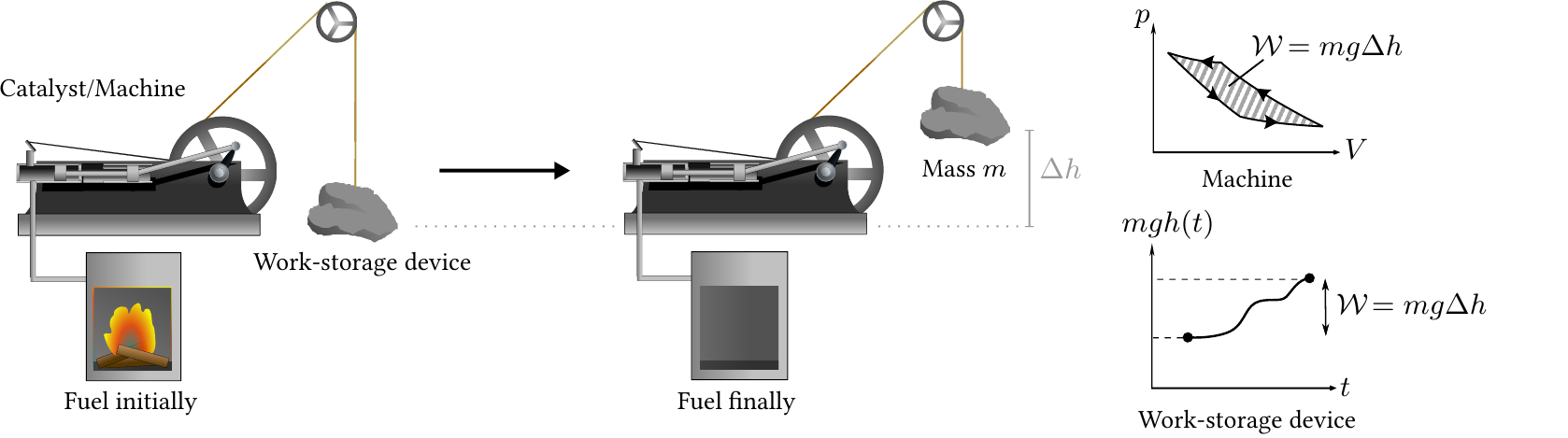}
\caption{\he{\emph{Left:} Phenomenological analogy of our setting. The catalyst corresponds to the machine that returns to its initial state, using up burning fuel to lift a weight. The burning fuel corresponds to a non-equilibrium system and the lifted weight corresponds to Arthur's work-storage device. \emph{Right:} }Work from two points of view. Path-dependent work obtained by looking at the time-dependent thermodynamic state of the thermal machine at the top and operational path-independent work obtained by looking at the weight (work-storage device) at the bottom. \he{All processes happen at some background-temperature $T$. The work of transition $W_{\mathrm{trans}}$ of the fuel corresponds to the maximal height that the weight can be lifted by arbitrary machines leaving the fuel in the corresponding final state and operating at background-temperature $T$.}}
\label{fig:fig1}
\end{figure*}

Since the notions and the use of language may be unfamiliar in the quantum thermodynamic
context, we will now specify clearly what we mean by free operations in the context of a resource theory.
Here, Arthur and Merlin have free access to heat baths at inverse temperature $\beta$. This will be relevant for the choice of the function $\mc{W}$, since Arthur will not pay a positive amount for a transitions that can be performed by only employing free resources. That is to say, it is important to specifically characterise the transitions that can be performed without expending valuable resources and only using baths.

Concretely, we assume that both Arthur and Merlin can pick heat baths, that is, quantum systems $B$ 
prepared in a Gibbs state
\begin{equation}\label{eq:Gibbsstate}
\omega_{H_B,\beta} = \frac{\exp(-\beta H_B)}{Z_H},
\end{equation}
with arbitrary Hamiltonian $H_B$. They can also apply any global unitary $U$
that commutes with the total Hamiltonian $H + H_B$. We use the short-hand
$H_A+H_B := {H_A\otimes\id} + {\id\otimes H_B}$ whenever the support of two
operators is clear from the context.
This amounts to the formalism of \emph{thermal operations} first introduced in Ref.\
\cite{Nanomachines}. It is also meaningful to allow for more general sets of
operations such as the so-called \emph{Gibbs-preserving maps}
\cite{Janzing00,Faist14,Wilming15}, also see the appendix, or simply \emph{thermalising maps} where the only possible interaction with the thermal bath is to bring the system to the Gibbs state in the spirit of Refs.\ \cite{Negative,Abergtrullywork,Anders2013}. Undoubtedly, in many thermodynamic settings, the latter one is the most realistic one capturing actual experimental situations.
The final form of the work quantifier $\mc{W}$ will in principle depend on
which model of operations with the bath is chosen, but the formalism
is general enough to be applicable in any of these situations. In the appendix we discuss in detail which are the minimal properties of the free operations that are explicitly needed to derive our results and show that the examples presented above have such properties.

More importantly, we will assume that both Arthur and Merlin, in addition to the heat bath, can also borrow any quantum ancillary system uncorrelated with the bath and the work-storage device, as long as it is returned in the same initial state and also uncorrelated with the work-storage device (see Fig.~\ref{fig:free_ops}). This ancillary system is referred to as a \emph{catalyst}, and its usage extends the set of transitions that can be performed with a bath \cite{SecondLaw,Ng14}. Such catalytic operations have been frequently studied in the recent literature of quantum thermodynamics, and naturally capture ``bystanders'', so auxiliary systems that 
may help performing transformations. In the following, we will refer to the operations described in this section as
\emph{free operations} when done without catalyst and \emph{catalytic free
operations} when performed with catalyst. Similarly, we will refer to the transitions induced by free and catalytic free operations as \emph{free transitions} and \emph{catalytic free transitions}, respectively. Lastly, given any object $p$, we define $\mc{F}(p)$ as the set of objects that can be reached from $p$ by free operations. Similarly, $\mc{F}_C(p)$ denotes the set of objects that can be reached from $p$ by catalytic free operations.
\begin{figure}
\includegraphics{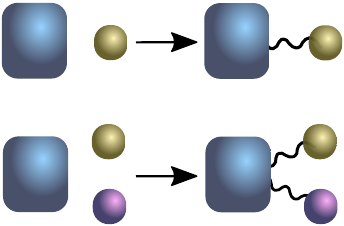}
\caption{Free operations and catalytic free operations. The big blue object denotes a heat bath, the yellow object a local system and the purple object a catalyst. The wiggly lines indicate the possible correlations after a {free operation (top figure), and {catalytic free operation} (bottom figure)}.}
\label{fig:free_ops}
\end{figure}

\subsection{The thermodynamical reading of the operational framework}
\label{sec:notionsofwork}

The game between Arthur and Merlin that we have introduced encodes a typical situation in the study of thermal machines, but does not describe it in the canonical way. The canonical analysis in the literature describes a thermal machine as composed by at least the following elements (see also Fig.\ \ref{fig:fig1}). 
\begin{itemize}
\item[i)] A heat bath at inverse temperature $\beta$, 
\item[ii)] a system $M$ out of equilibrium, \emph{i.e.} not in the Gibbs state \eqref{eq:Gibbsstate} at the temperature of the heat bath. {We will refer to this system as the \emph{fuel}, because is the resource that allows one to extract work,}
\item[iii)] a system $A$ where the work is stored. {This system is referred to as \emph{work-storage device},}
\item[iv)]  and possibly but not necessarily a catalyst $C$. 
\end{itemize}
Our game formulates the problem of evaluating the work given a transition of $A$, {that is, evaluating how much work has been stored in the work-storage device.} In the language of our game, $M$ would be any system that Merlin employs inside his machine performing the transitions {and that is possibly modified (burnt) after a protocol of work extraction}, i.e.,  {it plays the role of the ``fuel'' in traditional thermodynamics}. Such three-partite (four-partite, if the catalyst is explicitly included) structure is indeed the one followed in Refs.\
\cite{Horodecki13,Aberg,Piotr14,Faist14,Aberg,Piotr14,Dahlsten11,Negative,EgloffRenner,QuantitativeLandauer,
Abergtrullywork,Popescu2013,Popescu2013b,Malabarba14,Lostaglio14} where the function $\mc{W}$ is defined \emph{a priori} as related to the energy difference in various ways. For instance, in Refs. \cite{Popescu2013,Popescu2013b} {$\mc{P}$ is taken as the quantum analogue of a lifted weight} and the work quantifier is defined as
\begin{equation}\label{eq:workmean}
\mc{W}_{\text{mean}}( p_A^{(i)} \rightarrow p_A^{(f)} ) = \tr ( \rho_A^{(f)} H_A^{(f)} ) -\tr ( \rho_A^{(i)} H_A^{(i)} ). 
\end{equation}

{Another relevant model is the case of the \emph{$\epsilon$-deterministic work}, following the approach in Ref.\ \cite{Horodecki13}. It also gives rise to an instructive 
example} of how the constraints on the sets of allowed objects $\mc{P}$ come into play. The systems that are considered useful 
in that context are qubits such that 
\begin{equation}
\mc{P_{\epsilon}}:= \{ (\rho,H) \:  | \:  H=\Delta \proj{1}, \: \norm{\rho -\ketbra{E}{E}}\leq 2\epsilon \}
\end{equation}
where $\ket{E}$ is an  eigenvector of $H$, $\norm{ \cdot }$ is a norm on quantum states and $\epsilon < \frac{1}{2}$. The restriction $\mc{P}_{\epsilon}$ encodes that Arthur is interested in having states of well-defined energy or at least $\epsilon$-close to it. Work is then given by the energy difference of the closest energy-eigenstates. Formally as
\begin{equation}\label{eq:workepsilon}
\mc{W}_{\text{det}} (p_A^{(i)} \rightarrow p_A^{(f)})=f(\rho_A^{(f)},H_A^{(f)})-f(\rho_A^{(i)},H_A^{(i)}),
\end{equation}
with the function $f$ being defined \cite{footnotewbit} as
\begin{equation}
f(\rho,H)=
\begin{cases}
\Delta & \mbox{if } \norm{\rho -\proj{1}} < 1\\
0 & \mbox{if }  \norm{\rho -\proj{0}}<  1 \\
\end{cases}.
\end{equation}
As heuristically discussed in the introduction, both $\mc{W}_{\text{mean}}$ and $\mc{W}_{\text{det}}$ {and the limitations that they impose $\mc{P}$ can be regarded as particular cases of the general framework that we put forward. Nonetheless, they serve to illustrate the mathematical objects we are concerned with. The precise way they are recovered as particular cases and the subtleties that emerge when doing so will be discussed in Appendix \ref{sec:otherforms}}.

\subsection{Work of transition, work cost and work value}
It is important to distinguish $\mc{W}$ as a work quantifier on the work-storage device from other quantities that are usually analysed in thermodynamics and referred to as \emph{work}. {Once $\mc{W}$ is defined, one is then usually concerned, in the language of the present work,  with the optimal amount of work that Merlin can obtain by performing a transition on his system $M$. This quantity has been considered in the context of single-shot work extraction in Refs.\ \cite{Renes14,EgloffRenner,Faist15}. Here we will refer to this function as \emph{work of transition}.}

{\begin{definition}[Work of transition]\label{def:workoftransition} Given a work quantifier $\mc{W}$ and inverse temperature $\beta$, a set of restrictions $\mc{P}$, and 
initial and final objects of $M$, denoted by $p^{(i)}_M$ and $p^{(f)}_M$, respectively, the work of transition $W_{\mathrm{trans}}(p^{(i)}_M \rightarrow p^{(f)}_M,\beta)$ is defined as 
\begin{eqnarray}
\nonumber&&W_{\mathrm{trans}}(p_M^{(i)} \rightarrow p_M^{(f)},\beta) \\
&:=&\!\!\!\!\!\!  \sup_{ \substack{p_{A}^{(i)},p_{A}^{(f)} \in \mc{P} ; \\ p_{M}^{(f)} \otimes p_{A}^{(f)}  \in  \mc{F}_C(p^{(i)}_M \otimes p^{(i)}_A)\\  
}} \!\!\!\!\!\!\!\!  \mc{W}( p_A^{(i)} \rightarrow p_A^{(f)},\beta).\label{eq:workext}
\end{eqnarray}
\end{definition}
In \eqref{eq:workext}, we have introduced the short hand notation
\begin{equation}
p^{(i)}_M \otimes p^{(i)}_A:=\left(\rho^{(i)}_M \otimes \rho^{(i)}_A, H_{M}\otimes \id_A + \id_{M} \otimes H_A\right).
\end{equation}
\he{Also recall that the set $\mc{F}_C(p^{(i)}_M \otimes p^{(i)}_A)$ is determined by all those free transitions} that Merlin can perform, including interactions with a heat bath and the catalyst.}

{Notice that in contrast to $\mc{W}$, $W_{\mathrm{trans}}$ is evaluated on transitions on $M$ and \emph{not} on the work-storage system. In fact, the work associated to a given physical process can never be specified as a transition on $M$ alone. That is, it is impossible to determine $\mc{W}$ as a function of $p_M^{(i)}\rightarrow p_{M}^{(f)}$. One needs to either specify a particular catalytic free transition leading to $p_{M}^{(f)} \otimes p_{A}^{(f)}  \in  \mc{F}_C(p^{(i)}_M \otimes p^{(i)}_A)$, or to simply define it by taking the optimal one, as we do in \eqref{eq:workext}. It is precisely in this sense that work, as a function of transitions on $M$, is said to be a path dependent quantity when evaluated in transitions on $M$, and a path-independent quantity when evaluated in transitions on $A$. This is also the case in phenomenological thermodynamics: work can be specified by knowing only the initial and final height of the lifted weight, however it is path-dependent as function of the other components operating the machine, for example, the fuel (see Fig.\ \ref{fig:fig1}).}

{One may be tempted to think at this point that the distinction between the work of transition $W_{\mathrm{trans}}$ and the work quantifier at the work-storage device $\mc{W}$ is somehow artificial. 
In the end, both quantities evaluate transitions on systems and which player plays the role of Arthur or Merlin might seem at first sight arbitrary. However, let us insist that by no means Arthur and Merlin play an equivalent role. The key point is to understand that the transitions on Arthur systems are restricted so that $p_{A}^{(i)},p_{A}^{(f)}\in \mc{P}$. However, transitions on Merlin are fully unrestricte\ro{d.
% (Again, also in this sense, Merlin takes a similar role as he takes in common interactive proof systems. This asymmetry should not be surprising, since it is embedded in any thermodynamical framework.
 T}his is possibly most transparent in phenomenological thermodynamics: there, the work storage device undergoes a transition between two states of definite energy (a weight), however, the ``fuel'' employed in the process may undergo arbitrary transitions.}

{Lastly, note that in \eqref{eq:workext} we demand that the final state of $MA$ is uncorrelated. As discussed in Ref.\ \cite{Mueller14}, the creation of correlations can be a resource for performing thermodynamical transitions. Indeed, those correlations between $MA$ will turn out to play an important role in our axiomatic formulation. Hence, for 
reasons of clarity of presentation, we will first consider the case where no correlations are allowed, as in \eqref{eq:workext}, and study extensively the role of correlations in Sec.\ref{sec:secondlaw}.}

{Yet other relevant quantities in thermodynamics are given by the so-called \emph{work value} and \emph{work cost} defined as
\begin{eqnarray}
\label{eq:defworkvalue}&&W_{\mathrm{value}}(p_{M},\beta):=  W_{\mathrm{trans}}(p_M \rightarrow \omega_{\beta},\beta), \\
\label{eq:defworkcost}&&W_{\mathrm{cost}}(p_{M},\beta):=- W_{\mathrm{trans}}(w_{\beta} \rightarrow p_M ,\beta),
\end{eqnarray}}
where $w_\beta$ is an object describing a thermal state. The quantity $W_{\mathrm{value}}$ plays a relevant role since the second law is usually put as a bound on it. It describes how much work can be extracted from the system if it is viewed as a resource. In this work we are, however, mainly concerned with the form that $\mc{W}$ can take given a set of axioms. Clearly, the quantities $W_{\mathrm{trans}}$, $W_{\mathrm{value}}$ and $W_{\mathrm{cost}}$ can only be defined once $\mc{W}$ has been specified. We will show in Sec.\  \ref{sec:secondlaw}, though, that from the general properties of $\mc{W}$ implied by the axioms, we can find a second law \he{as $W_{\mathrm{value}}\leq W_{\mathrm{cost}}$}.

\section{Two basic axioms}

We are now in the position to formulate the basis on which all of the following analysis rests. We 
introduce two operational axioms concerning the work quantifier $\mc{W}$. They seem as
innocent as they are natural, and clearly capture features that any reasonable function of the above type quantifying work should satisfy. They 
are physically very intuitive. In order to precisely develop our operational framework, they will be formulated in 
the mindset of the game played by Arthur and Merlin, in the language of a so-called interactive
proof system. In this language, they simply ensure that none of the players 
can get arbitrarily rich without expending valuable resources. We will carefully discuss all 
implications of our results, however, also in a physical language, stressing that  the conclusions
we draw indeed give rise to a natural framework for naturally grasping concepts of work in 
quantum thermodynamics.

\begin{axm}[{Cyclic transitions of the work-storage device}]\label{prin:one}
For any cyclic sequence of transitions {of the work-storage device} $p_A^{(1)}\rightarrow p_A^{(2)}\rightarrow\cdots\rightarrow p_A^{(n)}=p_A^{(1)}$, such that 
$p^{(i)} \in \mc{P} \:\: \forall i$, the sum of the work-values of the individual transitions is larger than or equal to zero,
\begin{equation}
\sum_{i=1}^{n-1} \mc{W}(p^{(i)}_A\rightarrow p^{(i+1)}_A,\beta) \geq 0.
\end{equation}
\end{axm}
According to our convention, if $\mc{W}$ takes a negative value, then Arthur is benefiting from the transaction, \emph{i.e.} Merlin pays to Arthur. Hence, the previous axiom ensures 
that---taking the simplest case---Arthur cannot get rich by demanding Merlin to first do a transition $p^{(1)}\rightarrow p^{(2)}$ and then 
asking from him to undo the transition. If this principle was violated, Arthur could get infinitely rich just by repeatedly interacting with Merlin. Note that Arthur is not even himself implementing the transition, hence, he is by definition not expending any resource.

We will now impose our second axiom, which ensures in turn that Merlin cannot get arbitrarily rich without spending resources. 
\begin{axm}[{Cyclic transitions of the fuel}]\label{prin:two}
For any cyclic sequence of transitions {of the fuel (Merlin's system) $p_M^{(1)}\rightarrow p_M^{(2)}\rightarrow\cdots\rightarrow p_M^{(n)}=p_M^{(1)}$, 
the sum of the optimal work that Merlin can obtain in each sequence (this is given by $W_{\mathrm{trans}}$ in \eqref{eq:workext}) is smaller or equal to zero,
\begin{equation}\label{eq:axiom2}
\sum_{i=1}^{n-1} W_{\mathrm{trans}}(p_M^{(i)}\rightarrow p_M^{(i+1)},\beta) \leq 0.
\end{equation}}
\end{axm}
{According to our convention and \eqref{eq:workext}, if $W_{\mathrm{trans}}$ takes a positive value, then Merlin is benefiting from the transaction, \emph{i.e.} Arthur pays to Merlin. Hence, the previous axiom ensures that---taking the simplest case---Merlin cannot get rich by the overall process of burning his fuel in a transition $p_{M}^{(1)} \rightarrow p_{M}^{(2)}$ and then restoring the fuel back to its original state $p_{M}^{(1)} \rightarrow p_{M}^{(2)}$. If this was violated, that is, the l.h.s. of \eqref{eq:axiom2} was positive, then Merlin would get rich while not having burnt any fuel. It is important to notice that the objects $p_M^{(i)}$ of Axiom \ref{prin:two} are \emph{not} restricted to be in $\mc{P}$, since these restrictions apply to the work-storage device, and here we are concerned with transitions on the fuel (Merlin's systems) which are fully unrestricted. However, Axiom \ref{prin:two} depends on $\mc{P}$ because $W_{\mathrm{trans}}$ is defined as a function of $\mc{P}$ and $\mc{W}$, as given by \eqref{eq:workext}. Also, note that Axioms \ref{prin:one} and \ref{prin:two} allow---and this will be indeed the case as discussed in Sec.\ \ref{sec:irreversibility}---for the l.h.s. of eq.\ \eqref{eq:axiom2} to be strictly smaller than zero.}

As a final remark, note that both Axioms \ref{prin:one} and \ref{prin:two} encode the spirit of the second
law of thermodynamics: By preventing any of the two players to become arbitrarily rich without spending resources, we are enforcing the impossibility
to create a \emph{perpetuum-mobil\'e}. Our approach is, however, inverse to
what usually found in phenomenological thermodynamics. There, work is defined
\emph{a priori} through the lifted weight and the second law is understood as a
constraint on the possible physical processes. In contrast, Axioms \ref{prin:one} and \ref{prin:two} do not impose any constraint on the allowed physical operations that Merlin is performing. They simply state that one does \emph{not} account as work what can be generated with a bath and a catalyst with the \emph{a priori} given physical operations. As such, in our set-up it is also \emph{impossible to violate the second law}: If by using, say, a forthcoming post-quantum theory, someone claimed to extract work from a single heat bath, then it simply means---regardless of the details of such theory---that what it is referred to as work does not fulfil our operational principles.
\section{General properties of the work quantifier}\label{sec:generalproperties}

It is the key feature of the framework developed here that very basic principles already allow {one to 
formulate} stringent properties of possible work functions $\mc{W}$. In
this subsection we will turn to discussing properties of a work-function $\mc{W}$ that respects Axioms
\ref{prin:one} and \ref{prin:two}. For conceptual clarity, we will keep the discussion rather informal in this subsection. For a mathematically detailed and rigorous treatment, we refer to the appendix. Nevertheless, we will have to introduce some notation and definitions first.
We are looking for a function $\mc{W}$ that assigns a real number to any pair of objects $p^{(i)}=(\rho^{(i)},
H^{(i)})$ and $p^{(f)}=(\rho^{(f)},H^{(f)})$ that belong to the given set $\mc{P}$. Such an
inclusion is assumed throughout the remaining, unless explicitly mentioned otherwise. We will use Latin letters
$p,q,r,s,\ldots$ to denote objects and denote the work-value of a transition
$p^{(i)}\rightarrow p^{(f)}$ as $\mc{W}(p^{(i)}\rightarrow p^{(f)},\beta)$ or simply $\mc{W}(p^{(i)} \rightarrow p^{(f)})$ if $\beta$ is clear from the context. If the Hamiltonian of the two objects is 
identical, we will also use the notation $\mc{W}(\rho^{(i)}\rightarrow \rho^{(f)})$.
Let us recall from previous sections that given any object $p$, we define $\mc{F}(p)$ and $\mc{F}_C(p)$ as the set of objects that can be reached from $p$ by free operations and catalytic free operations, respectively. Also, in the following we always assume the existence of the \emph{empty object} $\emptyset\in \mc{P}$. Physically it means that there is no work-storage device. Formally it is given by the state $1$ with Hamiltonian $0$ on the Hilbert-space $\CC$. It therefore fulfils $p\otimes \emptyset=\emptyset\otimes p=p$ for any $p\in\mc{P}$.

\begin{theorem}[Form of work quantifiers]\label{thm:generalform}
A function $\mc{W}$ respects Axioms~\ref{prin:one} and \ref{prin:two} if and only if it can be written as
\begin{equation}\label{eq:diffmonotones}
\mc{W}(p\rightarrow q) = M(q) - M(p),
\end{equation}
for a function $M$ such that $M(\emptyset)=0$ and that fulfils the following property:
\begin{itemize}
\item \emph{Additive monotonicity:} For all $p^{(1)},\ldots,p^{(m)}$ and $q^{(1)},\ldots,q^{(m)}$ in $\mc{P}$ such that $\bigotimes_{i=1}^{m} q^{(i)} \in \mc{F}_C(\bigotimes_{i=1}^{m} p^{(i)} )$
\begin{equation}\label{eq:additivemonotonicity}
\sum_{i=1}^{m} M( q^{(i)} ) \leq \sum_{i=1}^{m}  M (p^{(i)} ).
\end{equation}
\end{itemize}
\end{theorem}

In particular, the theorem implies that  work, \emph{as measured by the work-storage device}, is not path-dependent in the sense that
\begin{eqnarray}
\label{eq:reversibility1}
&&\mc{W}(p\rightarrow q) + \mc{W}(q\rightarrow s) = \mc{W}(p\rightarrow s),\\
\label{eq:reversibility2}&&\mc{W}(p\rightarrow q)=-\mc{W}(q\rightarrow p)
\end{eqnarray}
and that no work can be extracted in a free catalytic transition,
\begin{equation}
\label{eq:noworkfromfreeoperations}
\mc{W}(p\rightarrow q) \leq 0,\quad \forall q\in \mc{F}_C(p).
\end{equation}
Thus the work-storage device can be treated similarly to the case of a massive
body under the influence of a conservative force in classical mechanics:
There is a state-variable $M$ and its difference along a transition determines the work-value of the transition. Using catalytic free operations, which generalise the concept of putting a system in contact with a heat bath in phenomenological thermodynamics, this state-variable cannot be increased.

 {Let us highlight that condition \eqref{eq:reversibility2} is perfectly compatible with the well-known notion of irreversibility that emerges when considering notions of deterministic work in the spirit of Refs.\ \cite{Dahlsten11,Negative,EgloffRenner,QuantitativeLandauer,
Abergtrullywork,Gemmer}. That is, \eqref{eq:reversibility2} \emph{is compatible} with $W_{\mathrm{value}}<W_{\mathrm{cost}}$ and more generally with 
\begin{equation}\label{eq:irrevers1}
W_{\mathrm{trans}}(p _M^{(i)}\rightarrow p _M^{(f)}) \neq -W_{\mathrm{trans}}(p _M^{(f)} \rightarrow p _M^{(i)}).
\end{equation}
The validity of \eqref{eq:irrevers1} for any nontrivial set $\mc{P}$ and an extended} discussion on the implications of eq.\ \eqref{eq:reversibility2} are discussed in Sec.\  \ref{sec:irreversibility}. The exhaustive proof of Thm.\ \ref{thm:generalform} can be found in Appendix \ref{sec:appgeneralaxioms}.

Let us now discuss briefly the significance of \emph{additive monotonicity}, with the following lemma.
\begin{lemma}[Consequences of additive monotonicity]\label{lemma:aditivemonotonicitymeaning}
If a function $M$ fulfils additive monotonicity and $M(\emptyset)=0$, then it fulfils also the following properties.
\begin{itemize}
\item \emph{Monotonocity:} $M(q) \leq M(p)$ $\forall$ $p,q$ $\in$ $\mc{P}$, such that $q\in \mc{F}_C(p)$.
\item \emph{Additivity:} $M(p_A \otimes p_B) =M(p_A)+M(p_B)$ $\forall$ $p_A,p_B,p_A\otimes p_B$ $\in$ $\mc{P}$.
\item \emph{Positivity:} $M(p)\geq 0$ $\forall$ $p\in \mc{P}$.
\end{itemize}
\end{lemma}

Nonetheless, additive monotonicity is strictly stronger than demanding that $M$ fulfils the three
%four
conditions of previous Lemma. To see this, consider for example objects $p_A^{(i)},p_A^{(f)},p_B^{(i)},p_B^{(f)}$ $\in \mc{P}$, but with $p_A^{(i)} \otimes  p_B^{(i)} ,p_A^{(f)} \otimes p_B^{(f)} \notin \mc{P}$. 
Note that monotonicity and additivity do not apply to objects that are not in the set $\mc{P}$. The condition given by \eqref{eq:additivemonotonicity} implies that 
\begin{equation}
M(p^{(f)}_A)+M(p^{(f)}_B) \leq M(p^{(i)}_A) +M(p^{(i)}_B)  
\end{equation}
if $ p^{(f)}_A \otimes p^{(f)}_B \in \mc{F}_C (p^{(i)}_A \otimes p^{(i)}_B)$. However, this condition could not have been derived from the conditions of monotonicity and additivity, since they do not apply to objects that lie outside the set $\mc{P}$.

\section{Free energies as work quantifiers}
At this point a most natural question emerges: What are reasonable and natural
candidates for a work quantifier fulfilling Axioms \ref{prin:one} and \ref{prin:two}? 
Clearly, the set of valid functions $\mc{W}$ will crucially depend on the set of allowed states $\mc{P}$. Now, we will show that if the $\mc{P}$ is fully unrestricted, then the conditions simplify to the well-known notions of monotonicity and additivity. 

\begin{theorem}[Work qualifiers in the unrestricted case]\label{lemma:forallp}
If the set $\mc{P}=\{(\rho,H)\}$ is the set of all quantum states $\rho$ and Hamiltonians $H$, then a function $\mc{W}$ respects Axioms \ref{prin:one} and \ref{prin:two} if and only if it can be written as
\begin{equation}
\mc{W}(p \rightarrow q) = M(q)-M(p)
\end{equation}
for a function $M$ with $M(\emptyset)=0$ such that
\begin{eqnarray}
\label{eq:simplemonotonicity}&&M(q) \leq M(p) \:\: \forall \:\: q\in \mc{F}_C(p) ,\\
\label{eq:simpleadditivity} &&M(p\otimes q) =M(p)+M(q).
\end{eqnarray}
\end{theorem}
\begin{proof}
One directions follows directly from Thm.\ \ref{thm:generalform} and Lemma \ref{lemma:aditivemonotonicitymeaning}. That is, by using Thm.\ \ref{thm:generalform} we have that the function $M$ fulfils additive monotonicity. Hence, by Lemma. \ref{lemma:aditivemonotonicitymeaning} one sees that if fulfils also the two properties of Lemma \ref{lemma:forallp}. To show the inverse relation, it suffices to show that \eqref{eq:simplemonotonicity} and \eqref{eq:simpleadditivity} imply additive monotonicity. Indeed, by taking $p=\bigotimes_{i=1}^{n}p^{(i)}$ and $q=\bigotimes_{i=1}^{n}q^{(i)}$, we have that \eqref{eq:simplemonotonicity} implies 
\begin{equation}
M\left(\bigotimes_{i=1}^{n}q^{(i)}\right) \leq M\left(\bigotimes_{i=1}^{n}p^{(i)}\right)
\end{equation}
if $\bigotimes_{i=1}^{n}q^{(i)}=\mc{F}_C(\bigotimes_{i=1}^{n}p^{(i)})$. Additive monotonocity follows straightforwardly applying now \eqref{eq:simpleadditivity}.
\end{proof}

The previous Lemma implies that any function $M$ fulfilling the properties \eqref{eq:simplemonotonicity} and \eqref{eq:simpleadditivity} (when appropriately re-scaled to fulfil $M(\emptyset)=0$) can be used to build valid work quantifier for every possible set $\mc{P}$. Here we present a natural family of
such monotones.

\begin{theorem}[Work quantifiers from R\'enyi divergences]\label{thm:renyi}
The work quantifier $\mc{W}_{\alpha}(p \rightarrow q,\beta) = \Delta F^{\beta}_{\alpha} (q) - \Delta F^{\beta}_{\alpha}(p)$ with 
\begin{equation}
 \Delta F^{\beta}_{\alpha} \big(p=(\rho,H)\big)=\frac{1}{\beta}S_{\alpha}(\rho||w_{H,\beta}),
\end{equation}
where $S_{\alpha}$ is the quantum R\'enyi divergence \cite{Mueller-Lennert13,Wilde13}, for any $\alpha>0$, fulfils Axioms \ref{prin:one} and \ref{prin:two} for every set $\mc{P}$.
\end{theorem}
The proof of this statement follows from the fact that R\'enyi divergences 
satisfy \eqref{eq:simplemonotonicity} and \eqref{eq:simpleadditivity} \footnote{See Ref.\ \cite{Mueller-Lennert13} for an exhaustive analysis of R\'enyi divergences. Indeed, for Thm.\ \ref{thm:renyi} we only need that the function $S_{\alpha}(\cdot||\cdot)$ fulfils the data-processing inequality and additivity. The former implies $\Delta F_{\alpha}^{\beta}(\mc{F}_C(p)) \leq \Delta F_{\alpha}^{\beta}(p)$.}. From all the possible choices of $\alpha$, the case of $\alpha=1$ corresponding to the free energy based on the standard von-Neumann entropy plays a crucial role that will be discussed in Sec.\  \ref{sec:super-additivity}.

\section{The second law and irreversibility}\label{sec:secondlaw}

We will now turn to discussing the connection of our framework developed here
and quantitative second laws of thermodynamics that emerge from it.
In the language introduced here, such second laws are captured by the
 work-value of an object $W_{\mathrm{value}}$ being necessarily 
 smaller than or equal to  its work-cost $W_{\mathrm{cost}}$ (defined in \eqref{eq:defworkvalue} and \eqref{eq:defworkcost}, respectively). {As already discussed after the formulation of Axioms \ref{prin:one} and \ref{prin:two}, the axioms already impose that the definition of work must not allow for \he{either} of the players to get arbitrarily rich, which in spirit encodes the second law of thermodynamics. Indeed, this intuition can be made explicit by noting that Axiom \ref{prin:two}, if we take $p^{(1)}_M=p=p^{(3)}_M$ and $p^{(2)}=\omega$ (where $\omega$ is any thermal object), implies
 \begin{equation}
 0\geq W_{\mathrm{trans}}(p \rightarrow \omega) + W_{\mathrm{trans}}(w \rightarrow p)
 \end{equation}
 which together with \eqref{eq:defworkvalue} and \eqref{eq:defworkcost} imply the second law in the form
 \begin{equation}\label{eq:secondlawuncorr}
 W_{\mathrm{value}}(p)\leq W_{\mathrm{cost}}(p) . 
 \end{equation}}
 We will now discuss the exact conditions when we can expect to get a strict inequality, which is a phenomenon usually referred to as \emph{irreversibility} and that emerges in all the analyses of deterministic work (also called single-shot work extraction) \cite{Dahlsten11,Negative,EgloffRenner,QuantitativeLandauer,
Abergtrullywork,Gemmer}. We will see that this will depend crucially on the restrictions that are imposed over the work-storage device given by the set $\mc{P}$.
 
\subsection{Restrictions imply irreversibility}\label{sec:irreversibility}

{Let us first consider the case where no restrictions are imposed on the form of the work-storage device,} that is, $\mc{P}$ is the set of all pairs of states and Hamiltonians.  It is then maybe not surprising that \emph{reversibility} arises, in the sense that {$W_{\mathrm{value}}=W_{\mathrm{cost}}$} is true. The reason for this is simple: the best strategy that Merlin can employ to extract work from an object $p_M$ is just giving the system to Arthur. In this case the transition on $A$ is given by $\emptyset \rightarrow  p_{A}=p_{M}$ and thus the work is given simply by $M(p_M)$. The same is true in the case of the work-cost of the object. Merlin can just create $p_M$ by taking it from Arthur. Hence, summarizing, we see that if $\mc{P}$ is unrestricted we find
\begin{equation}
  W_{\mathrm{value}}(p_M)=M(p_M)=W_{\mathrm{cost}}(p_M) \quad\text{($\mc{P}$ unrestricted)}. \nonumber
\end{equation}

At a more heuristic level, we have seen that the tasks involving thermodynamical work become trivial when no restrictions are imposed on $\mc{P}$, since the entire process reduces to Merlin giving the physical system he possesses to Arthur. 
This comes with no surprise if we think about the analogue situation in phenomenological thermodynamics. If Merlin has the typical resource in classical thermodynamics, namely, some instance of ``burning fuel'', then he cannot 
simply give it to Arthur, expecting that the latter accepts it as a valid form of work. But this is only because in phenomenological thermodynamics, it is explicitly assumed that work comes in a very specific form:
This could, for example, be the height of a massive body in a potential when $\mc{P}$ is the set of work-storage devices described by a deterministic state-variable. Otherwise, if $\mc{P}$ was completely unrestricted, giving to Arthur simply the burning fuel as such---and also all the other parts of the machine---would be indeed the best strategy for Merlin. Any other strategy would involve interactions with a heat bath, which would necessarily decrease the value of the burning fuel as measured by any monotone function. 

That said, the limitations on the set $\mc{P}$, rather than being a technicality, impose the very conditions so that non-trivial thermodynamical processes take place: Merlin will now have to transform resources $p_{M} \notin \mc{P}$ into resources that are in $\mc{P}$, possibly at the prize of dissipating  the resource partially, which in turn yields irreversibility of the form $W_{\mathrm{value}}<W_{\mathrm{cost}}$. To illustrate this point, we will discuss in detail particular examples of restrictions.

\subsection{Examples of restrictions: Redefining $\epsilon$-deterministic work}\label{sec:epsilondet}

Let us discuss some examples of meaningful restrictions that can be imposed on the set of states and Hamiltonians 
and see how they led to irreversibility in the form of $W_{\mathrm{value}}<W_{\mathrm{cost}}$. In Sec.\ \ref{sec:notionsofwork} we have already briefly introduced the notion of $\epsilon$-deterministic work. Intuitively it describes the situation of work-storage devices which are almost in energy-eigenstates and where work is measured in terms of the energy-difference of these eigenstates. 
The original formulation of $\epsilon$-deterministic work introduced in Ref.\ \cite{Nanomachines} does not qualify for a valid work quantifier respecting Axioms \ref{prin:one} and \ref{prin:two}, as it {is discussed in Appendix \ref{sec:otherforms}. However, the notion of $\epsilon$-deterministic work can be naturally integrated in our formalism modifying slightly the function $\mc{W}_{\text{det}}$, while keeping the physical constraint $\mc{P}_\epsilon$}. Let us therefore show how the idea can be transferred into our setting and cast into a valid work-quantifier. Consider the following set of qubit work-storage devices,
\begin{equation}
\mc{P}_\epsilon := \left\{(\rho,H)\ \big{|}\ \norm{\rho - \ketbra{E}{E}}_1< 2\epsilon, H\ket{E}=E\ket{E} \right\}. \nonumber
\end{equation}
The operational meaning of $\epsilon\geq 0$ roughly is the optimal probability to be able to distinguish the state $\rho$ from an energy-eigenstate in a measurement. As work-quantifier we can choose any work quantifier that respects Axioms \ref{prin:one} and \ref{prin:two} for the set $\mc{P}_{\epsilon}$. We will analyse for simplicity the one induced by the von~Neumann free energy, that is $\mc{W}(p\rightarrow p',\beta)=\Delta F^\beta_1(p')-\Delta F^\beta_1(p)$. In the case of $\epsilon=0$, all states are energy eigenstates and, if the Hamiltonian does not change in a transition, the work-quantifier simply measures the energy-difference between the states before and after the transition. This coincides with the original definition of $\epsilon$-deterministic work given in Ref.\ \cite{Nanomachines} only when $\epsilon=0$ and $p$ and $p'$ have the same Hamiltonian. However, it recovers in spirit the notion of $\epsilon$-deterministic work in a way that is compatible with our axiomatic approach.

Let us now show that one indeed obtains irreversibility in this setting. This can be shown in the simplest case of $\epsilon=0$. It is implied by the results of Ref.\ \cite{Nanomachines} that if Merlin has a full-rank system described by $p_M$, it cannot be used to induce a transition on the work-storage device of the form $\proj{0}_A \rightarrow \proj{E}_A$. Hence, $W_{\mathrm{value}}(p_M)=0$ for $\mc{P}_{0}$. Nonetheless, $p_M$ may by a system arbitrarily far from equilibrium, hence it is necessary to spend resources to create it and $W_{\mathrm{cost}}(p_M)\geq \Delta F^\beta_1(p_M)>0$.

Furthermore, we expect the phenomenon of \emph{irreversibility} to emerge in numerous physically meaningful sets other than the $\epsilon$-deterministic work extraction. For instance, 
one may imagine restrictions on $\mc{P}$ that reflect work-storage devices whose Hilbert-space dimension is bounded by some finite number. Alternatively, one may consider one whose states' entropy or free energy is bounded from above. We expect that irreversibility emerges in any such setting for at least some systems, since Merlin will not be in general allowed to give his system to Arthur. The former will have to interact with the heat bath leading to unavoidable dissipation and irreversibility. We will leave the detailed investigation of such scenarios for future work.

% \begin{figure}
% \includegraphics[width=0.6\linewidth]{WorkFigure4.jpg}
% \caption{{Picture to be moved or removed}The processes extracting $W_{\mathrm{value}}(p_M)$ from $p_M$ and then creating $p_M$ for a total work of $ W^{\text{uncorr}}_{\mathrm{cost}}(p_M)$ can be arranged in a sequential way as depicted. This arrangement is one particular case of the one considered in the formulation of Axiom \ref{prin:two} and depicted in fig.\ \ref{fig:axiom2}. Hence, it follows from \eqref{eq:eqaxiom2} that $W_{\mathrm{value}}(p_M) \leq W^{\text{uncorr}}_{\mathrm{cost}}(p_M)$.}
% \label{fig:secondlaw}
% \end{figure}

\section{The role of correlations, the second law and super-additivity}\label{sec:super-additivity}

{We will now turn to discuss the role of correlations between the fuel (Merlin system $M$) and the work-storage device $A$ and the implications that it has for the characterisation of the work quantifier $\mc{W}$. To do this, let us first define a quantity similar to the work of transition in Def.\ \ref{def:workoftransition}, but where the fuel is allowed to establish correlations with the work-storage device.
\begin{definition}[Correlated work of transition]\label{def:correlatedworkoftransition} Given a work quantifier $\mc{W}$ and inverse temperature $\beta$, a set of restrictions $\mc{P}$, 
and initial and final objects of $M$, denoted by $p^{(i)}_M$ and $p^{(f)}_M$, respectively, the \emph{correlated work of transition} $W^{\mathrm{corr}}_{\mathrm{trans}}(p^{(i)}_M \rightarrow p^{(f)}_M,\beta)$ is defined as 
\begin{eqnarray}
\nonumber&&W^{\mathrm{corr}}_{\mathrm{trans}}(p_M^{(i)} \rightarrow p_M^{(f)},\beta) \\
&:=&\!\!\!\!\!\!  \sup_{ \substack{p_{A}^{(i)},p_{A}^{(f)} \in \mc{P} ; \\ p_{MA}^{(f)}  \in  \mc{F}_C(p^{(i)}_M \otimes p^{(i)}_A)\\  
}} \!\!\!\!\!\!\!\!  \mc{W}( p_A^{(i)} \rightarrow p_A^{(f)},\beta).\label{eq:workextcorr}
\end{eqnarray}
\end{definition}
Note that the only difference with Def.\ \ref{def:correlatedworkoftransition} is that the supremum is taken over protocols that allow the final state $p_{MA}^{(f)}$ to have arbitrary correlations. \he{We can also define the correlated work cost and value as}
\begin{eqnarray}
\label{eq:defcorrworkvalue}&&W^{\mathrm{corr}}_{\mathrm{value}}(p_{M},\beta):=  W^{\mathrm{corr}}_{\mathrm{trans}}(p_M \rightarrow \omega_{\beta},\beta), \\
\label{eq:defcorrworkcost}&&W^{\mathrm{corr}}_{\mathrm{cost}}(p_{M},\beta):=- W^{\mathrm{corr}}_{\mathrm{trans}}(w_{\beta} \rightarrow p_M ,\beta),
\end{eqnarray}
 It is to be expected that Axioms \ref{prin:one} and \ref{prin:two} are not sufficient 
 to capture the second law in the case where correlations are allowed. For instance Axiom \ref{prin:two} captures the idea that Merlin cannot get rich while returning the fuel to the same initial state. But in principle, \he{it} does not prevent Merlin from getting rich by (\he{despite} returning the fuel to the same state) establishing correlations between the fuel and the work-storage device. This is indeed the case: we can find a work quantifier $\mc{W}$ fulfilling Axiom \ref{prin:one} and \ref{prin:two} such that $W^{\mathrm{corr}}_{\mathrm{value}}(p_M) > W^{\mathrm{corr}}_{\mathrm{cost}}(p_M)$. \he{Although this might not be} a surprising result we include here a specific example, because it \he{illustrates} how the notion of super-additivity will come into play: the example relies on} the use of work quantifiers $\mc{W}(p\rightarrow p')=M(p')-M(p)$ such that $M$ is \emph{not} super-additive, where super-additivity means that $M(p_{AB})\geq M(p_A)+M(p_B)$, whenever $p_{AB},p_A,p_B\in\mc{P}$. 
 
  \he{Assume now} a bipartite state $p_{MA}$ and a monotone $M$ such that super-additivity is violated, that is, $M(p_{MA})<M(p_M)+M(p_A)$. Let us first look at $W^{\mathrm{corr}}_{\mathrm{cost}}(p_M)$. One particular protocol to create $p_M$ consists of Arthur having initially $p_{MA}$ and giving subsystem $M$ to Merlin, while keeping $p_A$. This particular protocol gives an upper bound to the work cost as
\begin{equation}\label{eq:2lawviolation1}
W^{\mathrm{corr}}_\mathrm{cost}(p_M)\leq M(p_{MA})-M(p_A).
\end{equation}
Secondly, we can lower bound $W_{\mathrm{value}}(p_M)$, simply by performing the obvious protocol where Merlin gives $p_M$ to Arthur, resulting in
\begin{equation}\label{eq:2lawviolation2}
W^{\mathrm{corr}}_\mathrm{value}(p_M)\geq M(p_M).
\end{equation}
Combining the fact that $p_{MA}$ violates super-additivity with eq.'s \eqref{eq:2lawviolation1} and \eqref{eq:2lawviolation2} results in a strict violation of the second law $W^{\mathrm{corr}}_\mathrm{value}(p_M)>W^{\mathrm{corr}}_\mathrm{cost}(p_M)$.

Let us now discuss the implications of this example. Suppose that Merlin would like to use the fact that $W^{\mathrm{corr}}_\mathrm{value}(p_M)>W^{\mathrm{corr}}_\mathrm{cost}(p_M)$ to become arbitrarily rich, or in other words, create a \emph{perpetuum-mobil\'e}. He can start by having initially $p_M$ and obtaining $W^{\mathrm{corr}}_{\mathrm{value}}(p_M)$. Then he will create again $p_M$, having paid $W^{\mathrm{corr}}_{\mathrm{cost}}(p_M)$ and thus resulting in an overall benefit. Note that $M$ {is returned to its original state after each cycle}, however it becomes correlated with the work-storage device. When Merlin repeats those processes, he will need fresh uncorrelated work-storage devices each time, devices
that end up all being correlated with Merlin's catalyst 
{and among them}. Hence, Merlin is getting arbitrarily rich without spending resources in the sense that he is not changing his system which behaves like a catalyst, but he does spend resources, because he is establishing correlations between $M$ and a new work-storage device at each cycle. In other words, Merlin is spending ``absence of correlations'', hence it seems natural that he can obtain benefit from it. A similar, but non-equivalent, effect has been discussed in Ref.\ \cite{Mueller14}, where the correlations are established among different parts of the catalyst. Thus, one possible viewpoint is to state that in order to account properly for resources, correlations cannot be created. Hence, the second law would take the form \eqref{eq:secondlawuncorr} which is indeed fulfilled for any work quantifier satisfying Axioms \ref{prin:one} and \ref{prin:two}.

{A complementary approach to capture the role of correlations is to take the opposite view: Correlating the catalyst with the work-storage device does not spend any resource and hence, it should be considered a valid operation. Furthermore, the work quantifier has to be modified accordingly to prevent from violations of the second-law (even if correlations are created) as given by $W^{\mathrm{corr}}_\mathrm{value}(p_M)>W^{\mathrm{corr}}_\mathrm{cost}(p_M)$. For this, we introduce a reformulation of Axiom \ref{prin:two} that accounts for correlations. We highlight that we do not regard this reformulation as being as fundamental as Axiom \ref{prin:two}. It only aims at capturing in a consistent way which are the valid work quantifiers if correlations are treated as a free resource, in the spirit of Ref.\ \cite{Mueller14}. }

{\begin{axm}[{Correlated cyclic transitions of the fuel}]\label{prin:three}
For any cyclic sequence of transitions {of the the ``fuel'' (Merlin's system) $p_M^{(1)}\rightarrow p_M^{(2)}\rightarrow\cdots\rightarrow p_M^{(n)}=p_M^{(1)}$, the sum of the optimal work that Merlin can obtain in each sequence when correlations with the work-storage device are allowed (this is given by $W^{\mathrm{corr}}_{\mathrm{trans}}$ in \eqref{eq:workextcorr}) is smaller or equal to zero,
\begin{equation}\label{eq:axiom3}
\sum_{i=1}^{n-1} W^{\mathrm{corr}}_{\mathrm{trans}}(p_M^{(i)}\rightarrow p_M^{(i+1)},\beta) \leq 0.
\end{equation}}
\end{axm}}
The intuition behind Axiom \ref{prin:three} is \he{similar }{to the one of Axiom \ref{prin:two}, with the only difference that Merlin is not allowed to become arbitrarily rich even by creating correlations with the work-storage devices. Imposing Axiom \ref{prin:three} has two important consequences. Firstly, one can easily show that if one makes use of Axiom \ref{prin:three}, then the usual second law is fulfilled, stated as
\begin{equation}\label{eq:usualsecondlaw}
W^{\mathrm{corr}}_{\mathrm{value}}(p_M) \leq W^{\mathrm{corr}}_{\mathrm{cost}}(p_M). 
\end{equation}
Secondly, allowing for correlations has consequences on the allowed work quantifiers $\mc{W}$. Taking the simplest case of $n=2$ and $p_{M}^{(1)}=p_{M}^{(2)}=p_M$, Axiom \ref{prin:three} implies that $\he{W_{\mathrm{trans}}^{\mathrm{corr}}(p_M\rightarrow p_M)}\leq 0$ $\forall p_M$. Combining this fact with Def.\ \ref{def:correlatedworkoftransition} one can easily see that $\mc{W}$, in order to respect Axioms \ref{prin:one} and \ref{prin:three} has to satisfy
\begin{equation}
\nonumber \mc{W}(p_A \rightarrow q_A) \leq 0
\end{equation}
for all $q_A,p_A$ in $\mc{P}$ such that $q_A \in \mc{F}^{\text{Corr.}}_{C}(p_A)$, where we define $\mc{F}^{\text{Corr.}}_{C}(p)$ \he{to be} the set of objects that can be reached from $p$ by using thermal baths and an ancillary system that is left, after the interaction with the bath, with the same marginal state and Hamiltonian, but possibly correlated with the system. We will refer to this transitions as \emph{correlated catalytic free transitions}. \he{It is easy to see} that $p_{A} \otimes p_B \in \mc{F}^{\text{Corr.}}_{C}(p_{AB}))$ for $p_A,p_B,p_{AB} \in \mc{P}$. Together with additivity, this implies that in order to respect Axioms \ref{prin:one} and \ref{prin:three}, the work quantifier is written as $\mc{W}(p\rightarrow p')=M(p')-M(p)$, where 
\begin{equation}\label{eq:super-additivity}
M(p_{AB})\geq M(p_A) +M(p_B).
\end{equation}}
As a consequence, the following is true:

\begin{theorem}[Von-Neumann free energy in the unrestricted case]
Under Axioms~\ref{prin:one} and \ref{prin:three}, from all the R\'enyi free energies, only the von~Neumann free energy
\begin{equation}
\Delta F^\beta_1(\rho,H) := \frac{1}{\beta}S(\rho||\omega_{H,\beta})
\end{equation}
remains to be a valid monotone to define a work-quantifier for \emph{arbitrary} sets $\mc{P}$ \he{(up to a constant). It gives rise} to a second law in the form
\begin{equation}
\he{W^{\mathrm{corr}}_{\mathrm{value}}(p_M) \leq \Delta F^\beta_1(p_M) \leq W^{\mathrm{corr}}_{\mathrm{cost}}(p_M). }
\end{equation}
\end{theorem}
Note that the von~Neumann free energy can also be written as
\begin{equation}
\Delta F^\beta_1(\rho,H) = F^\beta_1(\rho,H) - F^\beta_1(\omega_{H,\beta},H)
\end{equation}
with $F^\beta_1(\rho,H) = \tr(\rho H) - S(\rho)/\beta$. It therefore closely resembles the phenomenological free energy $U-TS$, or more precisely the exergy with respect to an environment of temperature $T=1/\beta$. We hence recover the statement that the maximum amount of work that can be extracted by a working system with access to a heat bath of temperature $T$ is given by the exergy of the working system with respect to the temperature $T$---but using reasoning very different from that of phenomenological thermodynamics.
It is also interesting to see that on a formal level in the framework developed here, the von-Neumann free energy does not arise from considering an asymptotic setting, but rather arises from the way \he{correlations are taken into account. }

We have seen that super-additivity and the von~Neumann free energy emerge naturally once we allow for the creation of correlations between the catalyst and the system. A similar result was obtained in Ref.\ \cite{Mueller14}, where it was shown that, for classical states, the change of von~Neumann free energy decides whether a transition between two objects is possible if multiple catalyst can be used, which can become correlated with each other, but not with the system. 

In the light of the previous discussions one might wonder whether super-additivity already singles out the von~Neumann free energy as the \emph{unique} valid monotone to define a work quantifier in the case of correlated catalysis. \he{This is true in} the case of vanishing Hamiltonians but otherwise unrestricted sets $\mc{P}$,
which we state in the following theorem.

\begin{theorem}[Von-Neumann free energy as a work quantifier for vanishing Hamiltonians]
\he{Consider the set of all finite-dimensional quantum states and the vanishing Hamiltonian $\mc{P}= \{(\rho,\id)\}$ and free operations given by thermal operations. Then }
the unique work quantifier \he{with continuous monotone $M$ and} fulfilling Axioms \ref{prin:one} and \ref{prin:three} is given, up to a constant factor, by
\begin{equation}
\mc{W}(p \rightarrow p',\beta)=\Delta F^\beta_1(p')-\Delta F^\beta_1(p) 
\end{equation}
where $\Delta F^\beta_1$ is the von~Neumann free energy.
\end{theorem}
\he{\begin{proof} 
Without loss of generality, consider the candidates for a work-quantifier defined as $M(\rho,\id):=\alpha(\log d(\rho) - f(\rho))$, where $d(\rho)$ is the dimension of the Hilbert-space of $\rho$, $\alpha$ is some positive constant and $f(\rho)$ is a yet unspecified continuous (on states of fixed dimension) function. We will show that $f$ has to be given by the von~Neumann entropy. Since  $S(\rho || \id_{d(\rho)} / d(\rho)) = \log d(\rho) - S(\rho)$ this implies the claim. Using additivity, super-additivity  we immediately obtain that $f$ has to be additive and sub-additive. From monotonicity under thermal operations we obtain that a) $f(U\rho U^\dagger) = f(\rho)$ for any unitary and b) $f(\sum_{i}p_i U_i\rho U_i^\dagger) \geq f(\rho)$ for any probability-distribution $p_i$ over unitaries $U_i$. Property a) implies that $f$ only depends on the eigenvalues of $\rho$ and is therefore equivalent to a function $\tilde{f}$ on probability distributions, which fulfills additivity and sub-additivity. Property b) implies that $\tilde{f}$ is \emph{Schur-concave}, i.e., can only increase under random permutations. In Ref.~\cite{Mueller2015} it has been shown that for probabiliy-distributions without zeros, such a function is of the form $\tilde{f}(p) = c H(p) + c_{d(p)}$, where $H$ is the Shannon-entropy, $c\geq 0$ and $c_{d_1d_2} = c_{d_1} + c_{d_2}$. By continuity, this form extends to arbitrary probability-distributions and we obtain $f(\rho) = c S(\rho)+c_{d(\rho)}$, where $S$ is the von~Neumann entropy. From $M(\id_d/d,\id_d)=0$, we obtain $c \log d + c_d = \log d$. This implies
\begin{align}
  M(\rho,\id) = c\alpha \left(\log d(\rho)-S(\rho)\right) = \alpha' \left(\log d(\rho)-S(\rho)\right), \nonumber
\end{align}
which finishes the proof.
\end{proof}}
\he{One might wonder whether the result could also hold in infinite-dimensional systems. However, in such systems the vanishing Hamiltonian does not have a well-defined thermal state for any temperature, so that it should not be considered as a physical Hamiltonian on such systems.} \je{To extend this result to more general classes of Hamiltonians constitutes an interesting open problem.} 
%
%\begin{conjecture}[Von-Neumann free energy as a work quantifier]
%If $\mc{P}$ is the set of all states and Hamiltonians, the unique work quantifier fulfilling Axioms \ref{prin:one} and \ref{prin:three} is given, up to a constant factor, by $\mc{W}(p \rightarrow p',%\beta)=\Delta F^\beta_1(p')-\Delta F^\beta_1(p)$ where $\Delta F^\beta_1$ is the von~Neumann free energy.
%\end{conjecture}
%
Importantly, if true, this does still \he{not} imply that one can only make use of the von~Neumann free energy as a work quantifier. There are many situations of physical relevance where the set of $\mc{P}$ is restricted, where one could still conceive other work quantifiers. Indeed, we \he{have seen in Sec.~\ref{sec:secondlaw}} how imposing constraints $\mc{P}$, rather than a technicality, is crucial to recover several commonly discussed regimes in which thermodynamics is expected to operate.

\section{Summary}

In this work, we have approached the subtle and much discussed question of how extend the notion of work in thermodynamics to the small scale, where fluctuations and quantum effects play a relevant role. {We have done so by distinctly shifting the mindset that is usually taken when considering notions of work. We deviate from the 
implicit assumption that work should necessarily be determined by the energy stored in the quantum analogue of a lifted weight.
Instead, we consider arbitrary classes of systems other than lifted weights, intended to realistically account  for the transitions that a quantum thermal machine is expected to perform. Within this extended family of systems, we take a strictly operational approach and pose the problem of identifying reasonable 
functions that evaluate the value of a given transition; these functions are supposed to have basic properties analogue to the \he{familiar notion of work in phenomenological} thermodynamics. These properties are stated 
in the form of strictly operational axioms that capture minimum reasonable conditions that 
meaningful work quantifiers are expected to satisfy. This is again a distinct deviation in mindset: We do not define quantities ad-hoc, but 
aim at clarifying those characteristic features that any work quantifier should fulfil, providing a general framework.}

Remarkably, simple and elementary as these axioms may appear, they provide sufficient mathematical structure
to give rise to surprisingly detailed and stringent properties that any function 
quantifying work has to fulfil, properties that can be rigorously derived from the axioms. 

{One of the advantages the formalism is that it is general enough to allow one to derive central concepts in thermodynamics without taking the definition of work as energy in a lifted weight as an \emph{a priori} given element. For instance, our generalised work quantifiers give rise to quantitative versions of the second law. 
Similarly, one can precisely discuss notions of irreversibility in this framework, in the sense that in order to obtain useful work, it is necessary to dissipate
the resources provided by the ``fuel'', concomitant to familiar notions in thermodynamics.}

{When the system is taken to be an analogue of a lifted weight in the quantum regime, our general framework recovers the usual definition of work as the energy difference as a particular case. At an more heuristic level, this can be summarised by the insight that the task of extracting work is nothing but the transfer of free-energy from an arbitrary system (the fuel) to another system which has to fulfil a set of given restrictions (the work-storage device). In the specific situation in which 
those restrictions are such the work-storage device is a lifted weight, then the free energy coincides with the energy.}

For coherence of the presentation, we have focused on work quantifiers in quantum thermodynamics in the main text. It should be clear, however,
that the technical results achieved are general enough to capture also other quantum resource theories, beyond the
quantum thermodynamic context. The arguments laid out in main text and the supplementary material clearly 
highlight the role that catalysts and their correlations play in such resource theories. 
Furthermore, our results show that there is a close connection between catalysis, the built-up of correlations, and of reversibility. In particular, we have shown in what precise way a restriction of the state-space of work-storage devices is necessary in order to obtain irreversibility.
{It is our hope that the approach taken here can be seen as a further invitation to revisit notions derived from classical thermodynamics
and taking an operational perspective when aiming at clarifying in what precise way they can be extended to the quantum regime.}

\section{Acknowledgements}
This work has been supported by the EU (SIQS, AQuS, RAQUEL), the ERC (TAQ), 
the 
Alexander von Humboldt-Foundation and 
the Studienstiftung des Deutschen Volkes. 

%\bibliography{workdef5.bib}

\begin{thebibliography}{47}
\bibitem [{foo({\natexlab{a}})}]{footnotewbit}%
{In Ref.\
  \cite{Horodecki13} it is assumed that $\rho_A^{(i)}$ is the ground state of
  $H_A$. Hence, we have $f^\epsilon(\rho_A^{(i)},H_A^{(i)})=0$. What we
  introduce here as $\he{\mc{W}_{\text{det}}}$ is a generalisation that follows
  in spirit, but lies outside the scope formally laid out in Ref.
  \cite{Horodecki13}}\BibitemShut {NoStop}%
\bibitem [{foo({\natexlab{b}})}]{footnotemeanwork}%
 {Strictly speaking, the
  mean energy does not fulfil the property of being positive for all states in
  $\mc{P}^{\text{mean}}$, as it is also demanded by Thm.\ \ref{lemma:forallp}.
  However, this can be tackled by simply taking $\mc{P}^{\text{mean}}$ so that
  the mean energy of the states is bounded from bellow. In this way, we can
  take the origin of energies (which does not alter the definition of work) in
  such a way all the states have positive energy.}\BibitemShut {Stop}%
\end{thebibliography}

\newpage
\clearpage
%\onecolumngrid
\appendix

\section{Scenario and definitions}

\subsection{Transitions and free transitions}\label{secapp:definitions}
Let us consider a pair of a quantum states and a Hamiltonian $p=(\rho,H)$. 
In the following we, will call such pairs \emph{objects} and denote the associated Hilbert space by $\mc{H}(p)$, 
which for most of this work is taken to be finite-dimensional. 

\begin{definition}[Transition]\label{def:transition} A transition is defined by a pair of objects $p^{(i)},p^{(f)}$ and an ordering between them. We will refer to a transition as $p^{(i)} \rightarrow p^{(f)}$.
\end{definition}

\begin{definition}[State transition]\label{def:statetransition} This is a transition in which the Hamiltonian remains constant. That is, if $(\rho^{(i)},H)\rightarrow (\rho^{(f)},H)$, we will refer to a state transition and denote it simply, if the Hamiltonian is clear from the context, by $\rho^{(i)} \rightarrow \rho^{(f)}$.
\end{definition}

{\begin{definition}[Sequence]\label{def:sequence} A set of $n-1$ transitions of the form $\{p^{(k)} \rightarrow p^{(k+1)}\}_{k=1}^{n-1}$ is referred to as \emph{sequence}. We will simply denote it by $p^{(1)} \rightarrow p^{(2)} \rightarrow \cdots \rightarrow p^{(n)}$.
\end{definition}}

Such transitions are to be interpreted, in the context of the present work, as changes on the system and state Hamiltonian of the battery of Arthur as implemented by Merlin. 

\begin{definition}[Free image]
A free image is a function $\mc{F}$ that maps $p^{(i)}$ and a parameter 
$\beta$ into sets of objects $\{p_k\}
%_k
=\mc{F}(p^{(i)},\beta)$. When $F$ is such that 
the Hamiltonian remains constant, that is, 
\begin{equation}
\mc{F}(\rho^{(i)},
%H,
\beta)=\{(\rho_k,H)\},
\end{equation}
we will refer to it as \emph{free state-image}.
\end{definition}

\begin{definition}[Free transition]A free transition is defined as any transition $p^{(i)}\rightarrow p^{(f)}$, 
where $p^{(f)}\in \mathcal{F}(p^{(i)},\beta)$. When the parameter $\beta$ is clear 
from the context, we will denote a free transition simply as $p^{(i)} \rightarrow \mathcal{F}(p^{(i)})$.
\end{definition}

\begin{definition}[Tensoring objects]Given two objects $p=(\rho,H)$ and $p'=(\rho',H')$, we define the tensor product 
\begin{equation}
	p\otimes p':=(\rho \otimes \rho',H \otimes \id_{\mc{H}(p')} + \id_{\mc{H}(p)} \otimes H').
\end{equation}
\end{definition} 
In the definition we explicitly indicated on which tensor-factor the identity maps act. In the following, 
we will omit such indications when the information is clear from the context.

\begin{definition}[Non-interacting objects]If an object based on a bipartite system of parts $A$ and $B$ 
has the form 
\begin{equation}
	p=(\rho_{AB},H_A\otimes \id_B+ \id_A\otimes H_B) 
\end{equation}
we refer to it as \emph{non-interacting} object. 
\end{definition}
Non-interacting objects are those objects on which we define a partial trace.

\begin{definition}[Partial traces]Given any two objects $p_{S}=(\rho_S,H_S)$ and $p_{|S}=(\rho_{|S},H_{|S})$, we define the trace $\tr_{|S}$ as an operator acting on objects $p$ of the form 
\begin{equation}
p=p_{S} \otimes p_{|S}=\big(\rho_S \otimes \rho_{|S},H_S \otimes \id_{|S}+\id_S \otimes H_{|S}\big),
\end{equation} 
such that $\tr_{|S}(p)=p_S$. We extend this definition to all \emph{non-interacting} objects by the partial trace on quantum states.
\end{definition}
At this point a remark about Hamiltonians is in order. When we consider non-interacting objects, the local Hamiltonians are not well-defined: We can always change their traces by adding a global zero of the form $(\lambda \id_A)\otimes \id_B - \id_A\otimes (\lambda \id_B)$ to the global Hamiltonian. Therefore we will from now call two Hamiltonian operators equivalent if they differ by a multiple of the identity, $H\sim H+\lambda \id$. For simplicity, we will, 
however, not indicate this in our notation and will just refer to the equivalence classes as Hamiltonians. We could also just fix the trace of the Hamiltonians. It will become clear later, why we do not follow this path. 
\begin{definition}[Catalytic free image]Given the free image $\mathcal{F}$, we define the \emph{catalytic free image} $\mathcal{F}_C$ as 
\begin{equation}
\mathcal{F}_C(p^{(i)},\beta):=\{ p \: | \:\: \exists \:\: q ;\: p \otimes q \in \mathcal{F}(p^{(i)}\otimes q,\beta) \}.
\end{equation} 
\end{definition}

\begin{definition}[Catalytic free transition]A catalytic free transition is defined as any transition $p^{(i)} \rightarrow p^{(f)}$ with $p^{(f)}\in \mathcal{F}_C(p^{(i)},\beta)$. When the parameter $\beta$ 
is clear from the context, we will denote a free state-transition simply as $p^{(i)} \rightarrow \mathcal{F}_c(p^{(i)})$.
\end{definition}

{\begin{definition}[Assisted transitions \he{and sequences}]\label{def:assistedtrans}
Two objects $p^{(1)},p^{(2)}$ form
 a \emph{transition assisted by $(c_1,c_2)$} if
 \begin{equation}
 p^{(2)}\otimes c_2 \in \mc{F}(p^{(1)}\otimes c_1 ,\beta),
 \end{equation}
\he{Now consider a sequence of transitions $p^{(i)}\rightarrow p^{(i+1)}$ for $i=1,\ldots,n-1$. If each transition is a free transition assisted by $(c_i,c_{i+1})$, respectively, we say the sequence is assisted by $(c_1,c_n)$.}
\end{definition}
In other words, an assisted sequence is a sequence on objects that can be performed by using free operations and an ancilla that is at the end uncorrelated with the system but might have changed its state. }

{\he{Although} we would like to keep this definition fully general, let us anticipate that $\{c_i\}_i$ are going to play the role of the fuel employed by Merlin, which enables (assists) a transition or sequence of transitions, by changing its state (by being burnt).}

\he{\begin{definition}[Free sequence]\label{def:freesequence}
We call a sequence assisted by $(c,c)$ a \emph{free sequence}.
\end{definition}}
{Following with the interpretation of $c$ as the fuel, a free sequence is then a sequence of transitions that can be implemented while not spending any \he{fuel.}}

\subsection{Basic assumptions on the free transitions}

In the main text we have focused on the \emph{resource theory} of \emph{a-thermality}, where the free operations are, loosely speaking, defined as the energy preserving joint operations on system and bath. These are mathematically characterized by the GP-maps, or strictly contained subsets of operations, such as the thermal operations. However, our results apply potentially to widely different resource theories defined by other classes of free operations, not motivated by the thermodynamic context. In this endeavour, we aim at 
contributing to the emerging understanding of general resource theories
\cite{Brandao15,Fritz,Coecke14}.
We state below the first assumptions on the free operations that are needed in order to derive the results of Sec.\  \ref{sec:generalproperties} in the main text, in particular Theorem \ref{thm:generalform} (restated as Theorem \ref{thm:diffmonotones} in this appendix). 
\begin{property}[Composability]\label{prop:composability}
If $p^{(3)} \in \mathcal{F}\big(p^{(2)},\beta\big)$ and $p^{(2)}\in \mathcal{F}\big(p^{(1)},\beta\big)$, then $p^{(3)}\in \mathcal{F}\big(p^{(1)},\beta\big)$.
\end{property}

\begin{property}[Swapping products]\label{prop:swapping}
Given an object of the form $p^{(1)} \otimes \ldots \otimes p^{(n)}$, then 
\begin{equation}
P(p^{(1)} \otimes \ldots \otimes p^{(n)}) \in \mc{F}(p^{(1)} \otimes \ldots \otimes p^{(n)},\beta), \,\forall \beta, 
\end{equation}
where $P$ permutes the labels $(1,\ldots,n)$ into $(\sigma(1),\ldots, \sigma(n))$. 
\end{property}
Note that Property \ref{prop:swapping} implies that the identity is a catalytic free transition, that is, $p \in \mc{F}_C(p,\beta)$ for all $\beta$. This follows since one can take as catalyst $q=p$ and perform a swap between the system and the catalyst. 

% \begin{property}[Free objects] \label{prop:freestate}There exists a set of \emph{free objects} $w_\beta$, such that for any object $p$, 
% \begin{equation}
% p\otimes w_\beta \in \mc{F}(p,\beta).
% \end{equation}  
% The set of \emph{free objects} is closed under tensor-products: For two free objects $w_\beta,w'_\beta$, the tensor-product $w_\beta\otimes w'_{\beta}=w''_\beta$ is again a free object.
% \end{property}
% We stress the parameter-dependence of the free-objects: A free object for some parameter-value $\beta$ will in general not be a free object for some other parameter-value $\beta'$.

% \hew{[[HW: do we actually need the free objects property??]]}

\begin{property}[Tracing as free operation]\label{prop:tracingfree}
For any subsystem $S$ of $A_1 ,\ldots, A_N$ of a product object, tracing out is in the free image. That is, 
\begin{equation}
\tr_{S}(p_{A_1}\otimes \ldots \otimes p_{A_N})\in \mc{F}(p_{A_1} \otimes \ldots \otimes p_{A_N},\beta).
\end{equation}  
\end{property}
%%%
In the case where 
\begin{equation}
	S=\cup_{i=1}^{N}A_i,
\end{equation}  
 the entire system is traced out. In this case we introduce the notation $\tr_{S}(p):=\emptyset$. 
 In this instance Proposition \ref{prop:tracingfree} is also fulfilled and we denote it by $\mc{W}(p\rightarrow \emptyset,\beta)\leq 0$. The object $\emptyset$ can be seen as the pair $(1,0)$ on $\mc{H}=\CC$. Note that it therefore fulfills $p\otimes\emptyset=p$ for every object $p$. It is therefore a free object independent of $\beta$.

The next Lemma will turn out to be very useful in the subsequent sections.
\begin{lemma}[Mapping time to space]\label{lemma:timetospace} {Suppose $\mc{F}$ fulfills properties \ref{prop:composability} and \ref{prop:swapping} and let $p \rightarrow p'$ be an assisted transition by $(c,c')$ and $q\rightarrow q'$ be an assisted transition by $(c',c'')$. Then the transition $p\otimes q\rightarrow p'\otimes q'$ is an assisted transition by $(c,c'')$. }
\begin{proof}
Note that by Definition \ref{def:assistedtrans} of assisted {transition} and Property \ref{prop:composability}, the transition $p_1\otimes c\rightarrow p_m\otimes c'$ is free. Therefore, also the transition $p_q \otimes q_1\otimes c\rightarrow p_m\otimes q_1\otimes  c'$ is free. An equivalent argument implies that $p_m\otimes q_1\otimes c' \rightarrow p_m\otimes q_n\otimes c''$ is also free transition. Composing these two transitions yields that  $p_1\otimes q_1\otimes c\rightarrow p_m\otimes q_n \otimes c''$ is also a free transition.
\end{proof}  
\end{lemma}

% \hew{A useful property concerning the interplay of the above properties and catalysis is the following one, which we will use without mentioning in the rest of the work.
% \begin{lemma}[Parallel composition] Let $p_1,q_1$ and $q_2$ be objects and suppose that $q_2\in \mc{F}(q_1,\beta)$. Then 
% $p_1\otimes q_2\in\mc{F}_C(p_1\otimes q_1,\beta)$.
% \end{lemma}
% }

\subsection{Work quantifiers}

Once we have specified the transitions and the free transitions, we will define a quantifier of the \emph{value} of a given transitions within the set of allowed work-storage devices $\mc{P}$. We will always assume that the empty object $\emptyset$ is an element of $\mc{P}$. 

\begin{definition}[Work quantifier]We define the work quantifier as a function $\mathcal{W}$ that maps a transition within $\mc{P}$ and parameter $(p^{(i)} \rightarrow p^{(f)}, \beta)$ into the real numbers.
If $\beta$ is clear from the context, we will simply write $\mathcal{W}(p^{(i)} \rightarrow p^{(f)})$.
\end{definition}

\section{General axioms}\label{sec:appgeneralaxioms}

We will now present the axioms \ref{prin:one} and \ref{prin:two} of the main text, restated in a more precise manner by making use of the mathematical definitions of Sec.\  \ref{secapp:definitions}. 

\begin{appaxm}[{Cyclic transitions of the work storage device}]\label{axm:cyclic}
Given a collection of objects of the work-storage device$\{p^{(1)},\ldots,p^{(n)}\}\subset\mc{P}$ such that  $p^{(n)} =p^{(1)}$, then
\begin{equation}\label{eq:workpositivecycle}
\sum_{i=1}^{n-1}\mathcal{W}\left(p^{(i)}\rightarrow p^{(i+1)},\beta \right) \geq 0.
\end{equation}
\end{appaxm}

Axiom \ref{axm:cyclic} ensures that if a set of states can be arranged in a cyclic sequence, the total work, given by the l.h.s.\ of \eqref{eq:workpositivecycle}, cannot be negative. Otherwise, Arthur, who receives at the end the same object he possessed at the beginning, can repeat the protocol an arbitrarily number of times and obtain an arbitrarily large benefit.

{Axiom \ref{prin:two} in the main text is however formulated in terms of $\mc{W}_{\mathrm{trans}}$. However, as this quantity is given as a function of $\mc{W}$ and $\mc{P}$ we can reformulate Axiom \ref{prin:two} as being directly expressed in terms of $\mc{W}$ for transitions of the work-storage device, which makes it a more comfortable formulation to work in the following proofs of this appendix.}

{\begin{appaxm}[Reformulation of ``cyclic transitions of the fuel'']\label{axm:freeop} 
Let $\{p_A^{(k)}\rightarrow q_A^{(k)}\}_{k=1}^{n-1}$ be a collection of assisted transitions of the work-storage device, assisted by $(c_k,c_{k+1})$ respectively, with $c_{n}=c_1$. Then
\begin{align}\label{eqapp:axiom2}
\sum_{k=1}^{n-1}  \mc{W}(p_A^{(k)}\rightarrow q_A^{(k)},\beta) \leq 0.
\end{align}
\end{appaxm}}
{Importantly, note that the objects $p_A^{(k)}$ and $q_A^{(k)}$ $\forall k$ in this formulation describe the work-storage device, contrary to the main text formulation of Axiom \ref{prin:two}. A schematic depiction of the transitions involved in this Axiom is given by Fig. \ref{fig:axiom2}. Now we will show, that \he{although} formulated in seemingly unrelated terms, both formulations are equivalent.}

{First, let us state a Corollary of Axiom \ref{axm:freeop} that will be useful in further proofs.
\begin{corollary}[Cyclic free sequences]\label{cor:corofaxiom2} Let $p_A^{(1)} \rightarrow p_A^{(2)} \rightarrow \cdots \rightarrow p_A^{(n)}$ a \emph{free sequence}, then,
\begin{equation}
\sum_{k=1}^{n-1}  \mc{W}(p_A^{(k)}\rightarrow p_A^{(k+1)},\beta) \leq 0.
\end{equation}
\end{corollary}
Corollary \ref{cor:corofaxiom2} follows simply by the definition of free sequence, which is a particular case of the conditions of Axiom \ref{axm:freeop}, in the case where $q_A^{(k)}=p_A^{(k+1)}$.}
\begin{figure}
\includegraphics{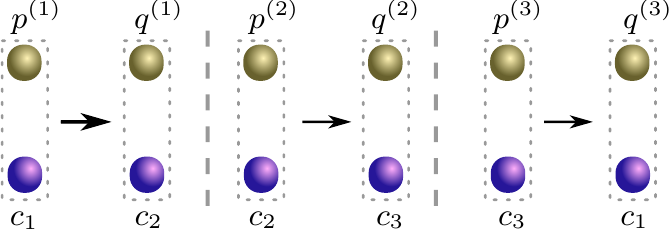}
\caption{A set of \he{transitions} where the constraints of Axiom \ref{prin:two} apply. Merlin holds a ancillary system described by $c_1$ and receives the work-storage device in $p^{(1)}$. Both systems are initially uncorrelated. Merlin performs a free transitions so that the final state is described by $c_2\otimes q^{(1)}$. After this process, Arthur comes with a new work-storage device, initially uncorrelated so that Merlin holds now $c_2\otimes p^{(2)}$. Several sequences of transitions are performed as depicted, so that at the final step, the ancillary system has returned to its initial state $c_1$. Importantly, also note that Merlin's ancillary system is uncorrelated with all the work-storage systems used in the process. Hence, Merlin, apart from the thermal baths which are considered free, has not spent any resource (neither in the form of changing his ancillary system, nor in the form of correlations) in the overall process. Axiom \ref{prin:two} states that Merlin cannot obtain benefit when adding up the work value of each transition.}
\label{fig:axiom2}
\end{figure}

\subsubsection{Equivalence between formulations:} {We will now show that indeed the version of the Axiom \ref{prin:two} given in the main text is equivalent to the one given above. }
{Let us first assume the version given in the main text. That is, we assume that for any sequence of the fuel (Merlin's system), where we $p_M^{(1)}\rightarrow \cdots\rightarrow p_M^{(n)}=p_M^{(1)}$, then
\begin{equation}\label{eq:eqequivalence1}
\sum_{i=1}^{n-1} W_{\mathrm{trans}}(p_M^{(i)}\rightarrow p_M^{(i+1)},\beta) \leq 0.
\end{equation}
Now consider a set of assisted transitions $\{p_A^{(k)}\rightarrow q_A^{(k)}\}_{k=1}^{n-1}$ of the work-storage device, assisted by $(c_k,c_{k+1})$ respectively, with $c_{n}=c_1$, as Axiom \ref{axm:freeop} states. Using Def. \ref{def:workoftransition}, we have that 
\begin{equation}
\mc{W}(p_A^{(k)} \rightarrow q_A^{(k)}) \leq W_\mathrm{trans}(c_k \rightarrow c_{k+1}).
\end{equation}
for all $k\in \{1,\ldots,n-1\}$. \he{But then by identifying} $c_i=p_M^{(i)}$ for all $i$, we obtain Eq.\ \eqref{eqapp:axiom2}.  }

{Let us now show the converse direction. We have to show, that given a sequence $p_M^{(1)}\rightarrow\cdots\rightarrow p_M^{(n)}=p_M^{(1)}$,  eq.~\eqref{eqapp:axiom2} implies eq.~\eqref{eq:eqequivalence1}. Each transition $p_M^{(k)} \rightarrow p_M^{(k+1)}$ will also induce a transition on the marginal of the work-storage device, given by $p_A^{(k)} \rightarrow q_A^{(k)} \in \mc{S}_k$, where $\mc{S}_k$ is the set of all marginal transitions on the work-storage device that can happen together with $p_M^{(k)} \rightarrow p_M^{(k+1)}$ on the fuel, and equivalently for all $k$. More explicitly, 
\begin{equation}
\mc{S}_k:= \{ p_A^{(k)} \rightarrow q_A^{(k)} \: | \:  q_A^{(k)} \otimes p_M^{(k+1)} \in \mc{F}_C (p_A^{(k)} \otimes p_M^{(k)})\}
\end{equation}
That said, all the $p_A^{(k)}\rightarrow q_A^{(k)}\in \mc{S}_k$ are an assisted transition by $p_M^{(k)}\rightarrow p_M^{(k+1)}$. By our assumption (eq.~\eqref{eqapp:axiom2}) this implies that the total work-value fulfills
\begin{equation}
\sum_{k=1}^{n-1}\mc{W}(p_A^{(k)}\rightarrow q_A^{(k)}) \leq 0
\end{equation}
for all $p_A^{(k)}\rightarrow q_A^{(k)}\in\mc{S}_k$. Then, this implies trivially
\begin{equation}\label{eq:eqequivalence_realization}
\sum_{k=1}^{n-1} \sup_{p_A^{(k)}\rightarrow q_A^{(k)}\in\mc{S}_k}\left(\mc{W}(p_A^{(k)}\rightarrow q_A^{(k)})\right) \leq 0,
\end{equation}
Now notice that the supremum in eq.~\eqref{eq:eqequivalence_realization} is the same as the one in the Def. \ref{def:workoftransition} of $W_{\mathrm{trans}}$, which concludes the proof.
 }

\subsection{Implications for the work definition}
We now turn to exploring implications for the work quantifiers. {Since the two Axioms have been reformulated in App. \ref{sec:appgeneralaxioms} in such a way that they only refer to objects of the work-storage device and not of the fuel, we will drop the labels $M$ and $A$. Unless explicitly mentioned, we will use the letters $p,q$ to refer to the work-storage device.}

\begin{lemma}[Properties of work quantifiers]\label{thm:3properties} Consider a free image $\mathcal{F}$ fulfilling Properties \ref{prop:composability}-\ref{prop:tracingfree}. In this case, Axioms \ref{axm:cyclic} and \ref{axm:freeop} are fulfilled if and only if $\mathcal{W}$ satisfies the following properties,
\begin{enumerate}
\item For all $p^{(1)},\ldots,p^{(m)}$ and $q^{(1)},\ldots,q^{(m)}$ in $\mc{P}$ such that $\bigotimes_{i=1}^{m} q^{(i)} \in \mc{F}_C(\bigotimes_{i=1}^{m} p^{(i)} )$, 
\begin{equation}\label{eq:additivemonotonocitywork}
\sum_{i=1}^{m} \mc{W}(p^{(i)} \rightarrow q^{(i)}) \leq 0.
\end{equation}
\item For all $p,q,r$ $\in \mc{P}$
\begin{eqnarray}
\label{eq:reversibility}&&\mc{W}(p\rightarrow q) = - \mc{W}(q\rightarrow p),\\
\label{eq:conservative}&&\mc{W}(p \rightarrow q)+\mc{W}(q \rightarrow r) = \mc{W}(p\rightarrow r).
\end{eqnarray}
%\item For all $p,q,r,s,p\otimes q,r\otimes s$ $\in \mc{P}$
%\begin{equation}
%\mc{W}(p\otimes q\rightarrow r\otimes s) = \mc{W}(p\rightarrow r) + \mc{W}(q\rightarrow s).
%\end{equation}
\end{enumerate}

\begin{proof}
We will first show that the axioms imply the properties, beginning with properties \eqref{eq:reversibility} and \eqref{eq:conservative}. The two properties follow immediately once we have shown that any cyclic sequence $p_1\rightarrow p_2\rightarrow\ldots\rightarrow p_n=p_1$ has a total work-value equal to zero. Given Axiom~\ref{prin:one}, {which already implies that it is larger than zero}, this only requires us to show that such a sequence has a work-value smaller or equal to zero. {This will be done by showing that any cyclic sequence is a free sequence, which is enough to show the claim given Corollary \ref{cor:corofaxiom2}. }We will show that any cyclic sequence is a {\emph{free sequence}, where following the notation of Def. \ref{def:freesequence}, $c_1=c_n:=c$ } is given by $c=\bigotimes_{i=2}^{n-1}p_i$. To see that $c$ assists any cyclic sequence from $p_1$ to $p_n=p_1$, consider the object $p_1\otimes c=\bigotimes_{i=1}^{n-1}p_i$. By swapping, which is a free operation, we arrive at state $p_2\otimes c'$ with $c'=p_1\otimes p_3\otimes p_4\otimes\cdots\otimes p_{n-1}$. Repeating the swapping sequentially we see that the first system goes through the transitions $p_1\rightarrow p_2\rightarrow \cdots\rightarrow p_{n-1}$. Applying a final swap the \he{fuel} is returned to $c$ and the system returns to object $p_1$, proving the claim and thus, eqs.\  \eqref{eq:reversibility} and \eqref{eq:conservative}. 

Let us now show property \eqref{eq:additivemonotonocitywork} from the axioms. The premise of \eqref{eq:additivemonotonocitywork} is that, there exists a catalytic free transition $\bigotimes_{i=1}^np_i\rightarrow \bigotimes_{i=1}^n q_i$. {\he{Here we are} taking $m=n$ without loss of generality\he{. The other cases follow by tensoring a suitable number of empty objects $\emptyset$}}. Then the transition $p_1\rightarrow q_1$ is {an assisted transition by}
\begin{equation}
\left(c_1=\bigotimes_{i=2}^n \:\: p_i,\:\: c'_1=\bigotimes_{i=2}^n\:\:q_i\right).
\end{equation}
Secondly, the transition $p_2 \rightarrow q_2$ is an {assisted transition} by $(c_2=\bigotimes_{i=2}^n\: q_i,\:p_2\bigotimes_{i=3}^n\:q_i)$. This can be seen by just performing a swap between the work-storage system in $p_2$ and the first element of the \he{fuel} in $q_2$. An equivalent swapping can be used to show that $p_j\rightarrow q_j$ is {an assisted transition} by
\begin{equation}
\left(c_j=\bigotimes_{i=2}^{j-1}p_i\: \bigotimes_{k=j}^nq_k,\:c_j'=\bigotimes_{i=2}^jp_i\bigotimes_{k=j+1}^nq_k\right)
\end{equation}
 for $j=3,\ldots,n-1$. Lastly, $p_n\rightarrow q_n$ is assisted by $(c_n=\bigotimes_{i=2}^{n-1}p_i\otimes q_n,c_n'=\otimes_{i=2}^np_i)$. Altogether, this implies that the set of sequences $\{p_i \rightarrow q_i \}_{i=1}^{n}$ can be each performed with free operations assisted by $(c_i,c'_i)$ as described previously. Note, that $c'_i=c_{i+1}$ and $c_1=c'_n$, hence, it meets the conditions of Axiom \ref{prin:two} which by eq.\ \eqref{eqapp:axiom2} implies
\begin{equation}
\nonumber \sum_{i=1}^n \mc{W}(p_i\rightarrow q_i)\leq 0.
\end{equation}

Finally, let us show that the properties \eqref{eq:additivemonotonocitywork}-\eqref{eq:conservative} imply the axioms. Axiom \ref{prin:one} is trivially satisfied since properties \eqref{eq:reversibility} and \eqref{eq:conservative} imply that for any cyclic sequence the total amount of work is zero. {Let us move to Axiom \ref{axm:freeop}, which has as a premise} that one has $n-1$ assisted transitions $p^{(j)}\rightarrow q^{(j)}$, assisted by $(c_j,c_{j+1})$ with $j=1,\ldots,n-1$ and $c_{n}=c_1$. Then, we can use Lemma~\ref{lemma:timetospace} and see that the transition 
\begin{equation}
\bigotimes_{j=1}^{n-1} p^{(j)} \rightarrow \bigotimes_{j=1}^{n-1} q^{(j)}
\end{equation}
is an assisted transition, assisted by $(c_1,c_{n}=c_1)$. {That is, the system $c$ is returned unchanged, hence $\bigotimes_{j=1}^{n-1} p^{(j)} \rightarrow \bigotimes_{j=1}^{n-1} q^{(j)}$ is indeed a catalytic free transition} and \eqref{eq:additivemonotonocitywork} implies that
\begin{equation}
\sum_{j=1}^{n-1} \mc{W}(p^{(j)}\rightarrow q^{(j)}) \leq 0,
\end{equation}
proving \eqref{eqapp:axiom2} and thus Axiom~\ref{prin:one}.
\end{proof}
\end{lemma}

Let us now show that Axioms \ref{axm:cyclic} and \ref{axm:freeop}, or equivalently eqs.\ \eqref{eq:additivemonotonocitywork}-\eqref{eq:conservative}, imply that the work function $\mc{W}$ must take a very particular form.

\begin{theorem} [Theorem \ref{thm:generalform} in the main text]\label{thm:diffmonotones}
Given a free image
$\mc{F}$ that fulfils Properties \ref{prop:composability}-\ref{prop:tracingfree}, the function $\mc{W}$ fulfils Axioms \ref{axm:cyclic} and \ref{axm:freeop} if and only if it can be written as
\begin{equation}\label{eq:diffmonotones}
\mc{W}(p\rightarrow q) = M(q) - M(p),
\end{equation}
for a function $M$ such that $M(\emptyset)=0$ and that fulfils the following property:
\begin{enumerate}
\item[] \emph{Additive monotonicity:} For all $p^{(1)},\ldots,p^{(n)}$ and $q^{(1)},\ldots,q^{(n)}$ in $\mc{P}$ such that $\bigotimes_{i=1}^{n} q^{(i)} \in \mc{F}_C(\bigotimes_{i=1}^{n} p^{(i)} )$
\begin{equation}\label{eq:additivemonotonicity}
\sum_{i=1}^{n} M( q^{(i)} ) \leq \sum_{i=1}^{n}  M (p^{(i)} ).
\end{equation}
\end{enumerate}
\end{theorem}

\begin{proof}
We will prove it by showing an equivalence with conditions \eqref{eq:additivemonotonocitywork}-\eqref{eq:conservative}, which in turn are equivalent with Axioms \ref{axm:cyclic} and \ref{axm:freeop}. Consider the function $M(p):= \mc{W}(\emptyset\rightarrow p)$. By properties \eqref{eq:reversibility} and \eqref{eq:conservative} we have
\begin{equation}
\mc{W}(p\rightarrow q) = M(q)-M(p)
\end{equation}
and $M(\emptyset)=0$ is true by definition. Clearly, \eqref{eq:additivemonotonocitywork} is fulfilled if and only if additive monotonicity holds for $M$. 
\end{proof}

\section{Gibbs-preserving and thermal operations}
\label{sec:app-GP}
In this section, we will turn to two classes of operations that can be used to model meaningful
classes of thermodynamic operations in the quantum regime, namely Gibbs-preserving operations (GPO) \cite{Janzing00,Faist14,Wilming15} and thermal operations (TO) \cite{Nanomachines}. We will first introduce the necessary objects, then define 
what state transitions are possible, and finally show that all the necessary 
properties are indeed fulfilled. 
Both GPO and TO have the same sets of free objects, induced by Gibbs-states.
\begin{definition}[Gibbs objects] The free objects of GPO and TO are given by
\begin{equation}
w = (\omega_H,H),\quad \omega_H = \frac{\exp(-\beta H)}{Z_H},
\end{equation}
with any Hamiltonian $H$, and called \emph{Gibbs objects}.
\end{definition}
Since % Gibbs objects are unitarily invariant in the sense that 
% \begin{equation}
% \omega_{UHU^\dagger}=U\omega_HU^\dagger 
% \end{equation}
% and 
$\omega_H=\omega_{H+\lambda \id}$ for any $\lambda \in \RR$, Gibbs objects 
are well-defined. To every object $p=(\rho,H)$ we can associate the Gibbs object 
%%%
\begin{equation}
w(p) = (\omega_H,H). 
\end{equation}
%%%

Let us now define Gibbs-preserving transitions.
\begin{definition}[Gibbs-preserving transition]
A transition $p=(\rho,H)\rightarrow q=(\sigma,K)$ is \emph{Gibbs-preserving} if there exists a quantum channel $\mc{G}$ such that
\begin{equation}
\sigma=\mc{G}(\rho)\ \text{and}\ \omega_K=\mc{G}(\omega_H).
\end{equation}
\end{definition}

Clearly, any Gibbs-object is mapped to another Gibbs-objects under Gibbs-preserving transitions, hence the name. In the case that the Hamiltonian $H$ does not change in a Gibbs-preserving transition we call the corresponding quantum channel $\mc{G}$ a Gibbs-preserving channel with respect to $H$. An operational way to think about the change of Hamiltonian in Gibbs-preserving transitions is given by Gibbs-preserving operations.
\begin{definition}[Gibbs-preserving operations] Any operation composed of taking thermal objects (at the fixed inverse temperature $\beta$), applying Gibbs-preserving channels and tracing out subsystems is called a \emph{Gibbs-preserving operation} (GPO). 
\end{definition}
GPO are closed under composition since the set of Gibbs-objects is closed under tensor products and a composition of two Gibbs-preserving channels is again Gibbs-preserving. Let us discuss some examples of GPO. A particular way to describe them is through maps from objects to objects which induce Gibbs-preserving transitions.

\begin{example}
Suppose $G$ maps objects to objects, such that
\begin{equation}
G(\rho,H) = (\mc{G}_H(\rho),\tilde{\mc{G}}(H))
\end{equation}
with $\mc{G}_H$ a quantum channel and $\tilde{\mc{G}}$ a map that maps Hamiltonians 
onto Hamiltonians. Furthermore, suppose that the maps $\mc{G}_H,\tilde{\mc{G}}$ fulfill the consistency relation
\begin{equation}
\mc{G}_H(\omega_H) = \omega_{\tilde{\mc{G}}(H)}
\end{equation}
for all objects $(\rho,H)$. Then $p\rightarrow G(p)$ is a Gibbs-preserving transition for any object $p$ and can be written as a Gibbs-preserving operation.
\end{example}
To see that the above construction can be seen as Gibbs-preserving operations let $\mc{T}_H$ be the channel which acts as $\mc{T}_H(\rho)=\omega_H$ for any $\rho$, $\mc{S}$ the swap channel $\mc{S}(\rho\otimes \sigma) = \sigma \otimes \rho$ and $\omega_H\otimes$ the channel that tensors in $\omega_H$, i.e., $\omega_H\otimes (\rho)=\omega_{H}\otimes \rho$. Now let $G$ be any Gibbs-preserving operation. It is easy to check from the consistency condition that $\mc{S}\circ \mc{G}_H\otimes \mc{T}_H$ is a Gibbs-preserving channel with respect to the Hamiltonian $H\otimes \id+\id\otimes \tilde{\mc{G}}(H)$. A simple calculation furthermore shows that on the level of quantum states we have
\begin{equation}
tr_2\circ \mc{S}\circ \left(\mc{G}_H\otimes\mc{T}_H\right)\circ \omega_{\tilde{\mc{G}}(H)}\otimes (\rho) = \mc{G}_H(\rho),
\end{equation}
while on the level of Hamiltonians we have the mapping $H\mapsto \tilde{\mc{G}}(H)$. 

In the following examples we use the notation of the previous example. 
\begin{example}
Suppose $\mc{G}_H={\rm id}$. Then $\mc{G}(H) = H$ (as equivalence classes) must hold true for the pair to be a Gibbs-preserving operation. Conversely, if $\mc{G}(H) = H$ then automatically 
$\mc{G}_H(\omega_H) = \omega_H$ has to be valid. 
\end{example} 
This example implies that $\mc{G}_H(\omega_H)=\omega_H$ if and only if $\mc{G}(H)=H$ (as equivalence class).
%%%%%
\begin{example}\label{app:ex:unitary-gpo}
Fix an $n$-dimensional unitary $U_n$ for every $n\in\NN$. Then letting $\tilde{\mc{G}}(H) = U_nHU_n^\dagger$ and $\mc{G}_H(\rho)=U_nHU_n^\dagger$, with $n\in\NN$ being
the dimension corresponding to the Hilbert space of $\rho$, defines a Gibbs-preserving operation.
\end{example}
This example implies that the \emph{swap-operation} $p\otimes q\mapsto q\otimes p$ is a Gibbs-preserving operation.
%%%5
\begin{example}
For any Gibbs object $w$, the map $p\mapsto p\otimes w$ is a Gibbs-preserving operation.
\end{example}
%%%%%
\begin{example}
For any non-interacting object $p_{A_1,\cdots, A_N}$ and any subset $S\subseteq\{A_1,\ldots,A_N\}$, the partial trace $p_{A_1,\ldots ,A_N}\mapsto p_S$ is a Gibbs-preserving operation.  
\end{example}
%%%%%%%
\begin{example}\label{app:ex:to} To every Hamiltonian $H$ choose a unitary $U_H$ such that
\begin{equation}
[U,H] = 0.
\end{equation}
Then the map $T$ which acts as
\begin{equation}
T(\rho,H) = \left(U_H \rho U_H^\dagger, H\right)
\end{equation}
is a Gibbs-preserving operation.
\end{example}
%%%%%
\begin{definition}[Thermal operations] 
A thermal operation is any operation that can be composed from the operations in the Examples~\ref{app:ex:unitary-gpo}--\ref{app:ex:to}.
\end{definition}
By definition, thermal operations are closed under composition.
%%%%%
\begin{definition}[Catalysis]
A transition $p\rightarrow r$ is a \emph{catalytic Gibbs-preserving transition} if there exists an object $q$ such that
\begin{equation}
p\otimes q \rightarrow r\otimes q
\end{equation}
is a Gibbs-preserving transition.
The transition is a \emph{catalytic thermal transition} if it is induced by a thermal operation.
\end{definition}
%%%%5
We will see later, Corollary \ref{cor:gibbstogibbs}, that also catalytic Gibbs-preserving transitions always map Gibbs objects to Gibbs objects. This is even true if correlations are allowed to built up between the system and the catalyst. 
\begin{proposition}
Gibbs-preserving transitions and transitions induced by thermal operations fulfil Properties~\ref{prop:composability}--\ref{prop:tracingfree}.
\begin{proof}
This immediately follows from the examples and the definition of thermal operations.
\end{proof}
\end{proposition}

\section{Free energy}
In this section we give an example for a valid work-quantifier if we choose as free state transitions Gibbs-preserving transitions or those induced by thermal operations. As customary, we define the von~Neumann \emph{free energy} of an object $p=(\rho,H)$ as the function

\begin{equation}
\Delta F_1^\beta(\rho,H) := F((\rho,H),\beta)-F((\omega_H,H),\beta),
\end{equation}
where
\begin{equation}
F(p,\beta) := \tr(\rho H) - \frac{1}{\beta}S(\rho),
\end{equation}
and $S(\rho)$ denotes the von~Neumann entropy of $\rho$.
We can also express $\Delta F_1^\beta$ using the quantum relative entropy, which is for states $\rho$ and $\sigma$ 
defined as
\begin{equation}
S(\rho || \sigma) = \tr\left(\rho \log \rho - \rho\log\sigma\right)
\end{equation}
if $\mathrm{supp}(\rho)\subseteq \mathrm{supp}(\sigma)$ and is equal to $\infty$ otherwise. It is well 
known that
\begin{equation}
\Delta F_1^\beta(\rho,H) = \frac{1}{\beta}S(\rho||\omega_H).
\end{equation}
This equation also directly shows that $\Delta F^\beta(\rho,H)$ is a well-defined function on objects: It is invariant under 
maps of the form
\begin{equation}
	H~\mapsto~H+\lambda \id 
\end{equation}
for any $\lambda\in \RR$.
In this section we will prove the following proposition.

\begin{proposition}[Properties of the von-Neumann free energy monotone]\label{app:thm:free-energy-monotone}
The function $\Delta F_1^\beta$ is a monotone under catalytic Gibbs-preserving transitions and fulfils
the following properties:
\begin{enumerate}
\item \emph{Normalisation:}
$\Delta F_1^\beta(w)=\Delta F_1^\beta(\emptyset)=0$ for any $w$ being a Gibbs object.
\item \emph{Extensivity:}
\begin{equation}
\Delta F_1^\beta(p_A \otimes p_B)=\Delta F_1^\beta(p_A) + \Delta F_1^\beta(p_B).
\end{equation}
\item \emph{Strong generalized super-additivity:} If $p^{(f)}_{AB},p^{(i)}_{AB}$ are non-interacting objects on $AB$,
\begin{eqnarray}
\Delta F_1^\beta(p^{(f)}_A) - \Delta F^\beta(p^{(i)}_A) &\geq& \Delta F_1^\beta(p^{(f)}_{AB})\\
&-&
\Delta F_1^\beta(p^{(i)}_{AB}) \nonumber
\end{eqnarray}
if $p^{(f)}_{AB}$ can be reached from $p^{(i)}_{AB}$ by only acting on subsytem $A$.
\end{enumerate}
\end{proposition}
Note that this proposition also implies that $\Delta F_1^\beta$ is a monotone for catalytic thermal transitions, since these 
constitute a strict subset of catalytic Gibbs-preserving transitions. Hence $\Delta F^\beta$ defines a valid work-quantifier for both Gibbs-preserving transitions and thermal operations. The property of strong generalized super-additivity furthermore implies the usual super-additivity $\Delta F_1^\beta(p_{AB})\geq \Delta F_1^\beta(p_A)+\Delta F_1^\beta(p_B)$, 
if $p_{AB}$ is non-interacting.

We will separate the proof into several propositions. We will frequently use the following well-known properties of the relative entropy: 
\begin{enumerate}
\item \emph{Positivity:} $S(\rho ||\sigma) \geq 0$ and $S(\rho || \sigma)=0$ if and only if $\rho=\sigma$,
\item \emph{Data-processing inequality}: $S(T(\rho)||T(\sigma)) \leq S(\rho || \sigma)$ for any quantum channel $T$.
\item \emph{Mutual information:}
For any bipartite state $\rho_{A_1A_2}$ we have
\begin{equation}
S(\rho_{AB}) = S(\rho_{A}) + S(\rho_{AB}) - S(\rho_{AB} || \rho_{A}\otimes \rho_{B}). 
\end{equation}
\end{enumerate}
Positivity directly implies that, for a fixed Hamiltonian $H$, the Gibbs-state $\omega_H$ at inverse temperature $\beta$ is the \emph{unique} minimum of the function $\rho\mapsto \Delta F_1^\beta(\rho,H)\geq 0$. Thus we already know that $\Delta F^\beta_1(w)=0$ for any Gibbs object and that $\Delta F_1^\beta(p)>0$ if $p$ is not a Gibbs object. 
\begin{proposition}[Extensivity of the free energy difference]
The function $\Delta F_1^\beta$ is extensive.
\begin{proof}
The proof
follows immediately from Property 3. of the relative entropy.
\end{proof}
\end{proposition} 
\begin{proposition}[Super-additivity of the free energy difference]
The function $\Delta F_1^\beta$ fulfils strong generalized super-additivity.
\begin{proof}
Assume that two objects
\begin{equation}
	p^{(i)}_{AB}=(\rho^{(i)}_{AB},H^{(i)}_A\otimes\id_B+\id_A\otimes H^{(i)}_B) 
\end{equation}	
	and 
\begin{equation}	
	p^{(f)}_{AB}=(\rho^{(f)}_{AB},H^{(f)}_A\otimes\id_B+\id_A\otimes H^{(f)}_B) 
\end{equation}	
are related through a local operation on $A$. Then 
\begin{equation}
\rho^{(f)}_{AB}=(\mc{T}_A\otimes \id)(\rho^{(i)}_{AB})
\end{equation}	 
for some quantum channel $\mc{T}_A$ acting on system $A$ and therefore $\rho^{(i)}_B=\rho^{(f)}_B$ and $H^{(i)}_B=H^{(f)}_B$. We need to show that
\begin{equation}\label{app:sgsa2}
\Delta F_1^\beta(p^{(f)}_A) - \Delta F^\beta(p^{(i)}_A) \geq \Delta F_1^\beta(p^{(f)}_{AB}) - \Delta F_1^\beta(p^{(i)}_{AB}). 
\end{equation}
If $\omega_{A},\omega_{B}$ are two Gibbs-states, then it is easy to prove, 
using locality, that for any state $\rho_{AB}$ on $AB$ we have
\begin{align}
S(\rho_{AB}||\omega_{A}\otimes \omega_{B}) = &S(\rho_{AB}||\rho_{A}\otimes \rho_{B}) \\ &+S(\rho_{A}\otimes \rho_{B}||\omega_{A}\otimes \omega_{B}).\nonumber
\end{align}
Using this relation we can rewrite the r.h.s. of eq.~\eqref{app:sgsa2} as
\begin{align}
&\frac{1}{\beta}\left[S(\rho^{(f)}_{AB}||\rho^{(f)}_A\otimes \rho^{(f)}_B) -  
S(\rho^{(i)}_{AB}||\rho^{(i)}_A\otimes \rho^{(i)}_B)\right] \\
&+\frac{1}{\beta}\left[S(\rho^{(f)}_{A}\otimes\rho^{(f)}_B||\omega_{H^{(f)}_A}\otimes \omega_{H^{(f)}_B})\right.\nonumber\\ &\quad-\left. S(\rho^{(i)}_{A}\otimes\rho^{(i)}_B||\omega_{H^{(i)}_A}\otimes \omega_{H^{(i)}_B})\right].\nonumber
\end{align}
Using $\rho^{(f)}_B=\rho^{(i)}_B$, $H^{(i)}_B=H^{(f)}_B$  and extensitivity, we find 
that the second term in brackets reduces to
\begin{equation}
\frac{1}{\beta}\left[S(\rho^{(f)}_A||\omega_{H^{(f)}_A}) - S(\rho^{(i)}_A||\omega_{H^{(i)}_A}) \right] = \Delta F_1^\beta(p^{(f)}_A) -\Delta F_1^\beta(p^{(i)}_A).
\end{equation}
But from the data-processing inequality we get that 
\begin{equation}
\frac{1}{\beta}\left[S(\rho^{(f)}_{AB}||\rho^{(f)}_A\otimes \rho^{(f)}_B) -  S(\rho^{(i)}_{AB}||\rho^{(i)}_A\otimes \rho^{(i)}_B)\right] = C \leq 0.
\end{equation}
We thus have
\begin{equation}
\mathrm{r.h.s.} = \Delta F_1^\beta(p^{(f)}_A) -\Delta F_1^\beta(p^{(i)}_A) + C \leq \Delta F_1^\beta(p^{(f)}_A) -\Delta F_1^\beta(p^{(i)}_A). 
\end{equation}
\end{proof}
\end{proposition}
What is left to be proven is that $\Delta F_1^\beta$ is a monotone under free (catalytic) transitions.

\begin{proposition}[Monotonicity under Gibbs-preserving transitions]
The function $\Delta F_1^\beta$ is a monotone under Gibbs-preserving transitions.
\begin{proof}
Consider a Gibbs-preserving transition
\begin{equation}
	p=(\rho,H)\rightarrow r=(\mc{G}(\rho),K)
\end{equation}
 with $\omega_K=\mc{G}(\omega_H)$. Then we get
\begin{align}
S(\mc{G}(\rho)|| \omega_K) &= S(\mc{G}(\rho)|| \mc{G}(\omega_H))\\
&\leq S(\rho ||\omega_H),\nonumber
\end{align}
where the last inequality is the data-processing inequality. 
\end{proof}
\end{proposition}
\begin{proposition}[Monotonicity under catalytic Gibbs-preserving transitions]
The function $\Delta F_1^\beta$ is a monotone under catalytic Gibbs-preserving transitions. 
\begin{proof}
Consider a catalytic transition $p\otimes q \rightarrow r\otimes q$. From monotonicity and extensitivity of $\Delta F_1^\beta$,
we obtain
\begin{align}
\Delta F_1^\beta(r) & = \Delta F_1^\beta(r\otimes q) - \Delta F_1^\beta(q) \\
&\leq \Delta F_1^\beta(p\otimes q) - \Delta F_1^\beta(q) \nonumber \\
&= \Delta F_1^\beta(p). \nonumber
\end{align}
\end{proof}
\end{proposition}
A similar proof can also be given in the setting where the catalyst is allowed to become correlated with the system. In this case, we need to use super-additivity of $\Delta F_{1}^{\beta}$. The same applies for the next corollary.
\begin{corollary}[Mapping Gibbs objects to Gibbs objects]\label{cor:gibbstogibbs}
Catalytic Gibbs-preserving transitions map Gibbs objects to Gibbs objects.
\begin{proof}
Consider a transition $w\otimes q\rightarrow r\otimes q$. Then $\Delta F^\beta(r)\leq \Delta F^\beta(w) = 0$. But $\Delta F^\beta\geq 0$ and $\Delta F^\beta$ vanishes only on Gibbs-objects. Hence $r$ has to be a Gibbs object.
\end{proof}
\end{corollary}
% The last two propositions also hold true when the catalyst becomes correlated with the system, that is, 
% during a Gibbs-preserving transition $p_{A_1}\otimes p_{A_2}\rightarrow q_{A_1A_2}$ such that $q_{A_2}=p_{A_2}$.  This follows from super-additivity of $\hew{\Delta F^\beta}$. 
This finishes the proof of Proposition~\ref{app:thm:free-energy-monotone}.

\section{The usual notions of work as a particular case in our formalism}\label{sec:otherforms}

{In this section we will review the common definitions of work that have been considered in the literature and recast them as particular cases of our formalism. That is, we will show that the energy stored in the work-storage device is a valid work quantifier fulfilling Axioms \ref{prin:one} and \ref{prin:two}, where the catalytic free operations and the set of restrictions $\mc{P}$ encode the behavior of a lifted weight.}

\subsection{The average energy of the lifted weight}

{Here we will discuss the model of work considered in \cite{Popescu2013,Popescu2013b}. In this case, the restrictions $\mc{P}$ are taken as the quantum analogue of a lifted weight:
\begin{equation}\label{eq:pmean}
\mc{P}^{\mathrm{mean}}=\left\{ (\rho_A,H_A) \: \: ; \: \: H_A =\he{mgX}\right\},
\end{equation}
\he{where $X$ is the position operator associated to one continuous degree of freedom, $m$ is the mass of the weight and $g$ is the gravitational constant.} Note that no restrictions are put on the state $\rho_A$ but the Hamiltonian $H_A$ is fixed throughout the protocol. The work quantifier is defined as
\begin{equation}\label{eq:workmean}
\mc{W}_{\text{mean}}( p_A^{(i)} \rightarrow p_A^{(f)} ) = \tr ( \rho_A^{(f)} H_A ) -\tr ( \rho_A^{(i)} H_A). 
\end{equation}
Importantly, the treatment of the work-storage device as a lifted weight is encoded in the set of catalytic free operations $\mc{F}_C$. Following the formalism of Ref.\ \cite{Popescu2013} and adding to it the notion of a catalyst, we have that the free operations, that we denote by $\mc{F}_C^{\text{mean}}$ are given by
\begin{equation}
\mc{F}_{C}^{\text{mean}}(\rho_A)=  \big\{ \rho'_A  \: \: ; \:\: \rho'=\tr_{BC} (U \omega_B \otimes \sigma_C \otimes \rho_AU^{\dagger})\big\},
\end{equation}
where $\omega$ is a Gibbs state (reflecting a heat-bath), the mean-energy is preserved
\begin{equation}
\tr \big((U \omega_B \otimes \sigma_C \otimes \rho_AU^{\dagger}-\omega_B \otimes \sigma_C \otimes \rho_A)H_{BCA})\big)=0,
\end{equation}
the catalyst $C$ is left in the same final state,
\begin{equation}
\sigma_C=\tr_{BA} (U \omega_B \otimes \sigma_C \otimes \rho_AU^{\dagger}),
\end{equation}
and finally, that the unitary $\he{U}$ commutes with the space-translation operator on $A$ (see \he{Refs.} \cite{Popescu2013,Malabarba14} for details). }

{Given all the conditions, it is shown in Refs. \cite{Popescu2013,Malabarba14} that 
\begin{equation}
\tr(\rho_{A} H_A) \leq \tr(\rho'_{A} H_A) \: \: \forall \: \: \rho'_A \in \mc{F}_{C}^{\text{mean}}(\rho_A).
\end{equation}
That is, the function $M(\rho,H):=\tr(\rho H)$ is a monotone under $\mc{F}_{C}^{\text{mean}}$ catalytic free operations. Lastly, one can easily show that the average energy fulfils the properties of additivity and \ro{super-additivity}, hence, using Thm.\ \ref{lemma:forallp} we see that $\mc{W}_{\text{mean}}$ fulfills Axioms \ref{prin:one} and \ref{prin:two} \cite{footnotemeanwork}.}

{As a final remark, notice that the free operations $\mc{F}_{C}^{\text{mean}}$ impose a limitation in comparison to what is usually allowed when thermal operations or Gibbs preserving maps are considered. The condition that the unitary has to commute with the translation operator of the work-storage device prevents one from employing the lifted weight as an entropy sink in the spirit of the example of Fig. \ref{fig:counterex}. At the same time, it is obvious from the model of the work-storage device and the conditions on  $\mc{F}_{C}^{\text{mean}}$ that this idealisation will not represent the realistic behaviour of a nano-machine. A work-storage device made of a few atoms certainly will not have a Hamiltonian of the form \eqref{eq:pmean} neither one can expect the operations performed in a real experimental device to, even approximately, commute with the translation \he{operator on the work-storage device, even if its Hilbert-space allows for such operators}.}

\subsection{The \emph{wbit} and $\epsilon$-deterministic work extraction}

{Now we will consider the model of a \emph{wbit} and the notion of $\epsilon$-deterministic work as it has been put forward in Ref.\ \cite{Nanomachines}. The restrictions on the work-storage device are such they are qubits with 
\begin{equation}
\mc{P_{\epsilon}}:= \{ (\rho,H) \:  | \:  H=\Delta \proj{1}, \: \norm{\rho -\ketbra{E}{E}}_1\leq 2\epsilon \},
\end{equation}
where $\ket{E}$ is an eigenvector of $H$, $\norm{ \cdot }_1$ is \he{the 1-norm} on quantum states and $\epsilon < \frac{1}{2}$. The restriction $\mc{P}_{\epsilon}$ encodes that Arthur is interested in having states of well-defined energy or at least $\epsilon$-close to it. Work is then given by the energy difference of the closest energy-eigenstates, formally as
\begin{equation}\label{eq:workepsilon}
\mc{W}_{\text{det}}(p_A^{(i)} \rightarrow p_A^{(f)})=f(\rho_A^{(f)},H_A^{(f)})-f(\rho_A^{(i)},H_A^{(i)}),
\end{equation}
with the function $f$ being defined \cite{footnotewbit} as
\begin{equation}
f(\rho,H)=
\begin{cases}
\Delta & \mbox{if } \norm{\rho -\proj{1}} < 1\\
0 & \mbox{if }  \norm{\rho -\proj{0}}<  1 \\
\end{cases}.
\end{equation}}

{It is easy to see that strictly speaking, this model of the \emph{wbit} respects Axioms \ref{prin:one} and \ref{prin:two} if $\epsilon=0$. In that case, $\mc{P}_{\epsilon=0}$ is given only by \he{pure energy eigenstates}. Hence, we find that $\mc{W}_{\mathrm{det}}$ evaluated on $\mc{P}_{\epsilon=0}$ coincides with the work-quantifier defined by the non-equilibrium free-energy on any transition where the Hamiltonian is constant. However, for $\epsilon>0$ we find that this model of the lifted weight does not satisfy the Axioms \ref{prin:one} and \ref{prin:two}. Indeed, one can simply check that it does not fulfil Eq.\ \eqref{eq:noworkfromfreeoperations}, or in other words, it is possible to store work in the \emph{wbit} by simply putting it in contact with a single thermal bath. Indeed, it has been shown in Ref. \cite{Wilming15} that for any value of $\epsilon>0$ and $\beta$ one can find a value of $\Delta>0$ such that there exists a thermal operation that brings a qubit work-storage system initially in the ground state to a final state $\rho^{(f)}=(1-\epsilon)\proj{1} + \epsilon \proj{0}$. Hence, for any $\epsilon > 0$, there exist $p^{(i)},p^{(f)} \: \in \:  \mc{P}_{\epsilon}$, such that $\he{\mc{W}_{\text{det}}(p^{(i)}} \rightarrow p^{(f)})=\Delta > 0$, while $p^{(f)} \in \mc{F}_C (p^{(i)})$, in contradiction with \eqref{eq:noworkfromfreeoperations}. Thus, $\he{\mc{W}_{\text{det}}}$ only defines a work quantifier that is compatible with Axioms \ref{prin:one} and \ref{prin:two} if $\epsilon=0$. A discussion on how to define a work quantifier that incorporates the notion of $\epsilon$-deterministic work without running into contradictions is presented in Sec.\  \ref{sec:epsilondet}.}

{As a final remark, let us note again that the incompatibility of $\epsilon$-deterministic work with our axioms is unrelated with issues related to reversibility or the fact that $\mc{W}(p\rightarrow q)=-\mc{W}(q\rightarrow p)$. Indeed, also $\epsilon$-deterministic work fulfils this property. The reason that makes it violate the axioms is the same as the one given in the example of Fig. \ref{fig:counterex}: when $\epsilon>0$, the \emph{wbit} can act as an entropy sink. Thus, identifying work with energy as $\mc{W}_{\text{det}  }$ does, allows one to extract work by using a single heat bath. Nonetheless, as we propose in Sec.\ \ref{sec:epsilondet} in the main text, it is possible to keep the spirit of the $\epsilon$-deterministic work (that the work-storage devices are $\epsilon$-close to pure energy eigenstates) and put forward a proper work quantifier that satisfies the axioms.}

\section{Probability distributions of work}
\label{sec:prob_distributions}

We will now discuss how our formalism is perfectly compatible with the notion of work as a classical random variable and the well-known results that pertain to the fluctuations of the probability distribution of work of Refs.\ \cite{Jarzynski97,Crooks99}. In the setting in the focus of attention in so-called fluctuation theorems one considers a system $S$ on which an energy measurement is performed both at the beginning and the end of a given unitary evolution. That is, the initial and final energies $E^{(i)}$ and $E^{(f)}$ are random variables, and so is the work given by $w=E^{(f)}-E^{(i)}$, \ro{which occurs with probability $P_W(w)$. Let us assume that the energy difference is bounded so that $P_W(w)\neq0$ only if $w_\text{min} \leq w \leq w_\text{max}$.} One may then always assume the presence of a work-storage device $A$ that stores the energy lost by $S$ and that is---in each event---in an energy eigenstate. 

More explicitly, we consider $\mc{P}=\{ (\proj{x},mgX)\}$ where $\proj{x}$ is an eigenstate of the \ro{truncated position operator $X=\int_{w_\text{min}}^{w_{\text{max}}} \proj{x} dx$ (taking $mg=1$ for simplicity) \footnote{\he{Here, $\int_{w_\text{min}}^{w_{\text{max}}} \proj{x} dx$ should be understood as a finite-dimensional Hamiltonian with non-degenerate spectrum within $[w_{\mathrm{min}},w_{\mathrm{max}}]$, as dense as necessary to reflect all the possible work-values.}}.} Then, in each event---that is, conditioned on a
specific value of the initial and final measurement---the work-storage device undergoes the transition 
\begin{equation}\label{eq:eigenstatetransition}
p^{(i)}=(\proj{0}_A,X_A) \rightarrow p^{(f)}=(\proj{w}_A,X_A)
\end{equation}
Then, we simply take
\begin{equation}\label{eq:workeigenstatetransition}
\nonumber \mc{W}(p^{(i)} \rightarrow p^{(i)}) =f(p^{(f)})-f(p^{(i)})=w
\end{equation}
 with $f((\proj{x}_A,X_A))=x$. Clearly, this work quantifier fulfills Axioms \ref{prin:one} and \ref{prin:two}, since it coincides, for the case of $\mc{P}=\{ (\proj{x},mgX)\}$, with taking  $f=\Delta F_{1}^{\beta}$, which by Thm.\ \ref{thm:renyi}  fulfils the axioms.

Nonetheless, note that it is crucial to include as part of the work extraction scheme the step in which a measurement is performed at the beginning and at the end. This is necessary even if $S$ is a classical system. Otherwise, the process will result in a mixed state of the work-storage device. 

To be more precise, it is important to appreciate that the \ro{two following notions are not equivalent:}
\begin{enumerate}
\item[i)]  $\nonumber \left(\proj{0}_A,X_A) \rightarrow (\proj{w}_A,X_A\right)$ occurs with probability $P_W(w)$.
\item[ii)]  $\nonumber \left(\proj{0}_A,X_A) \rightarrow (\sum_{w}P_W(w)\proj{w}_A,X_A\right)$ takes place.
\end{enumerate}

The interpretation of work as a random variable in Refs.\ \cite{Jarzynski97,Crooks99} corresponds to a process of the type i), where the probability distribution of work $P_W(w)$ encodes our a priori knowledge or capability to make predictions about which transition of the form \eqref{eq:eigenstatetransition} is going to take place. The transition given by ii) is a situation that is not covered by $\mc{P}=\{ (\proj{x},mgX)\}$ and $\mc{W}$ given simply by \eqref{eq:workeigenstatetransition}. In order to quantify work for a transition of the form ii), one has to properly account for the fact that the work-storage device might act as an entropy sink. In other words, in order to account for transitions of the form ii), one cannot identify work simply with the energy difference, and one has to define a work quantifier that fulfils Axioms \ref{prin:one} and \ref{prin:two} \ro{for extended sets $\mc{P}$ that contain $\sum_{w}P_W(w)\proj{w}_A$ as a valid state.}

Although this discussion between the differences of i) and ii) is rather obvious, we would like to stress that it plays an important role in the interpretation of work 
in quantitative terms, since i) requires to perform an energy measurement, which is not a free operation within any sensible thermodynamic framework. Hence, it should be kept in mind that when referring to work as a random variable, one is effectively \ro{quantifying} the work extracted/invested in the process plus the work extracted/invested in the measurements,
which of course is a perfectly valid approach.
In contrast, in the formalism of Refs.\ \cite{Aberg,Nanomachines}, deterministic values of energy are obtained, not by conditioning on the value of an energy measurement, but by engineering the protocol in such a way the work-storage device ends in a deterministic state of energy. 

Lastly, let us point out that our axiomatic framework can incorporate the notion of work as a random variable in situations more general than the one consider above. \he{For example,} 
one may consider probability distributions of work \he{for arbitrary processes between two measurements described by POVMs $\{M_\alpha\}_\alpha$, provided that all the post-measurement states are valid work-storage devices, i.e., are fairly included in the set $\mc{P}$.} Then, for any work quantifier $\mc{W}$ that fulfills Axioms \ref{prin:one} and \ref{prin:two} for the set $\mc{P}$, the work $\mc{W}(p^{(i)}_{\alpha_i}  \rightarrow p^{(f)}_{\alpha_f})$ occurs with probability $P(\alpha_i,\alpha_f)$, where $p^{(i)}_{\alpha_i}$ is the initial state conditioned on having obtained outcome $\alpha_i$ initially (and equivalently for \he{the state} $p^{(f)}_{\alpha_f}$ \he{after the final measurement}) and $P(\alpha_i,\alpha_f)$ is the joint probability distribution of obtaining the pair $\alpha_i,\alpha_f$. In this way, \he{we see that} the framework laid out in this work
 and the picture of capturing work as a probability distribution \he{are compatible}.

\clearpage

\end{document}